\documentclass[12pt]{article}
\usepackage[margin = 1in]{geometry}
\usepackage{amsmath}
\usepackage{amssymb}

\usepackage{amsthm}
\newtheorem{theorem}{Theorem}
\newtheorem{corollary}{Corollary}[theorem]

\usepackage{enumerate}
\usepackage{graphicx}
\usepackage{float}
\usepackage[labelfont=bf,justification=justified,format=plain]{caption}
\usepackage{subcaption}
\usepackage{tabularx}
\usepackage[flushmargin]{footmisc}
\usepackage{algorithm}
\usepackage{algpseudocode}
\usepackage{hyperref}
\usepackage{cleveref}
\crefname{lemma}{Lemma}{Lemmas}
\crefname{assumption}{Assumption}{Assumptions}
\crefname{algorithm}{Algorithm}{Algorithms}
\crefname{section}{Section}{Sections}
\crefname{corollary}{Corollary}{corollaries}
\crefname{appendix}{Appendix}{appendices}
\renewcommand{\footnoterule}{%
  \vspace{0.2in}
  \kern-4pt
  \hrule width 0.4\columnwidth
  \kern 2.6pt}
\usepackage{xcolor}
\usepackage[utf8]{inputenc}
\usepackage{amsmath}
\usepackage{amsfonts}
\usepackage{bbm}
\usepackage{mathtools}
\usepackage{tikz}
\usetikzlibrary{decorations.pathreplacing}
\usepackage{amsthm}
\usepackage[english]{babel}
\usepackage{fancyhdr}
\usepackage{amssymb}
\usepackage{qtree}
\usepackage{mathrsfs}
\usepackage{harvard}

\def\rrr#1\\{\par
\medskip\hbox{\vbox{\parindent=2em\hsize=6.12in
\hangindent=4em\hangafter=1#1}}}

\usepackage{crossreftools}
\let\ORGhypersetup\hypersetup
\protected\def\hypersetup{\ORGhypersetup}
\title{\textbf{Pre-validation Revisited}}
\author{Jing Shang$^{1}$ \quad Sourav Chatterjee$^{1,2}$ 
\quad Trevor Hastie$^{1,3}$ \quad Robert Tibshirani$^{1,3}$}
\date{\normalsize
$^1$Department of Statistics \\
$^2$Department of Mathematics \\
$^3$Department of Biomedical Data Science \\
Stanford University}

\begin{document}
\maketitle

\begin{abstract}
    Pre-validation is a way to build prediction model with two datasets of significantly different feature dimensions. Previous work showed that the asymptotic distribution of the resulting test statistic for the pre-validated predictor deviates from a standard Normal, hence leads to issues in hypothesis testing. In this paper, we revisit the pre-validation procedure and extend the problem formulation without any independence assumption on the two feature sets. We propose not only an analytical distribution of the test statistic for the pre-validated predictor under certain models, but also a generic bootstrap procedure to conduct inference. We show properties and benefits of pre-validation in prediction, inference and error estimation by simulations and applications, including analysis of a breast cancer study and a synthetic GWAS example.
\end{abstract}

\section{Introduction}
Modern biomedical technologies have transformed how we diagnose, treat, and prevent diseases in the past decade. In particular, we see surging needs in combining gene expression data with traditional clinical measurements. However, since gene expression data is usually high-dimensional, while the total number of available clinical measurements is rather limited, naive approaches such as pooling all features together will not work well. Particularly, the microarray features would dominate in inference and variable selection, resulting in negligence of valuable information from clinical data or failure of type I error control.

To address the issue caused by imbalance in feature dimensions, \citeasnoun{tibshirani2002pre} proposed the \textit{pre-validation} procedure to make a fairer comparison between the two sets of predictors. In the setting of gene expression data and clinical measurements, the procedure includes two steps: first we use high-dimensional gene expression data to train the leave-one-out fits of the response, then we use the fitted values and clinical data to build a final prediction model. It turns out that the pre-validation procedure not only enables us to test whether the gene expression data have predictive power with type-I error control, but also gives a good estimate of the prediction error. \autoref{fig:preval} illustrates the two-stage pre-validation procedure.

\begin{figure}[H]
    \centering
    \includegraphics[scale=0.8]{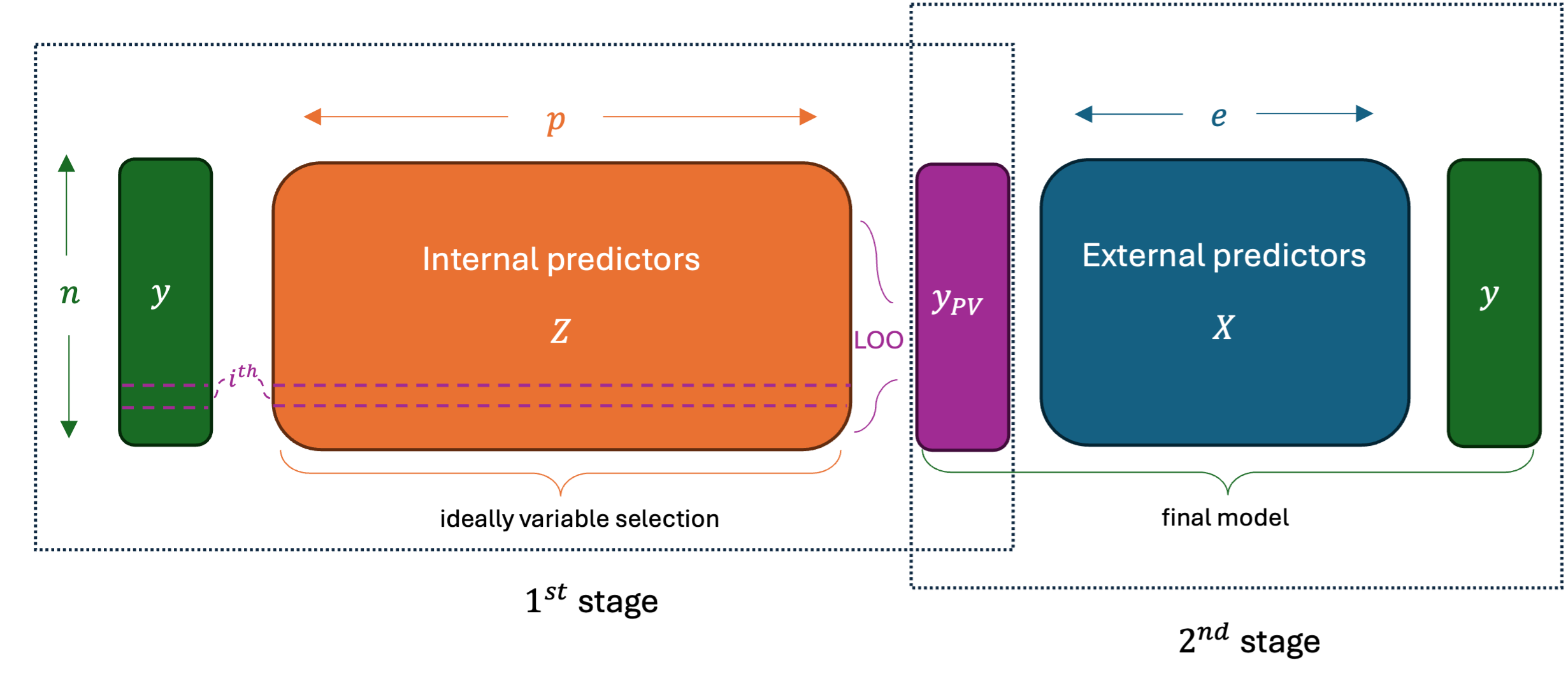}
    \caption
    {\em{Illustration of the Pre-validation Procedure:}
    \small first we use gene expression data $Z$ to train the leave-one-out fits $y_{PV}$ of response $y$, then we use the fitted values and clinical data $X$ to build a final prediction model.}

    \label{fig:preval}
\end{figure}

In the usual regression setting, we know that the test statistic of a coefficient asymptotically follows a standard Normal distribution under the null, and use this fact to conduct hypothesis test on the coefficients. However, \citeasnoun{hofling2008study} have shown that t-statistic for testing significance of the pre-validated predictor does not follow Normal distribution asymptotically. Instead, they provided a corrected asymptotic distribution in the two-stage linear regression procedure, when internal predictors (i.e. genetic features) and external predictors (i.e. clinical features) are independent. 

\autoref{fig:lsnoX} shows an example. It plots the quantiles of standard Normal and their proposed distribution against the test statistics, when there are no external predictors in the model. We can see that the distribution of test statistics have a heavier tail than the standard Normal. It gives an example of how the use of a standard Normal distribution to test hypotheses can cause issues. For more sophisticated learning models such as ridge or Lasso regression, they did not provide an analytical asymptotic distribution for the test statistic, but rather proposed a permutation test.

\begin{figure}[H]
     \centering
     \begin{subfigure}[b]{0.48\textwidth}
         \centering
         \includegraphics[width=\textwidth]{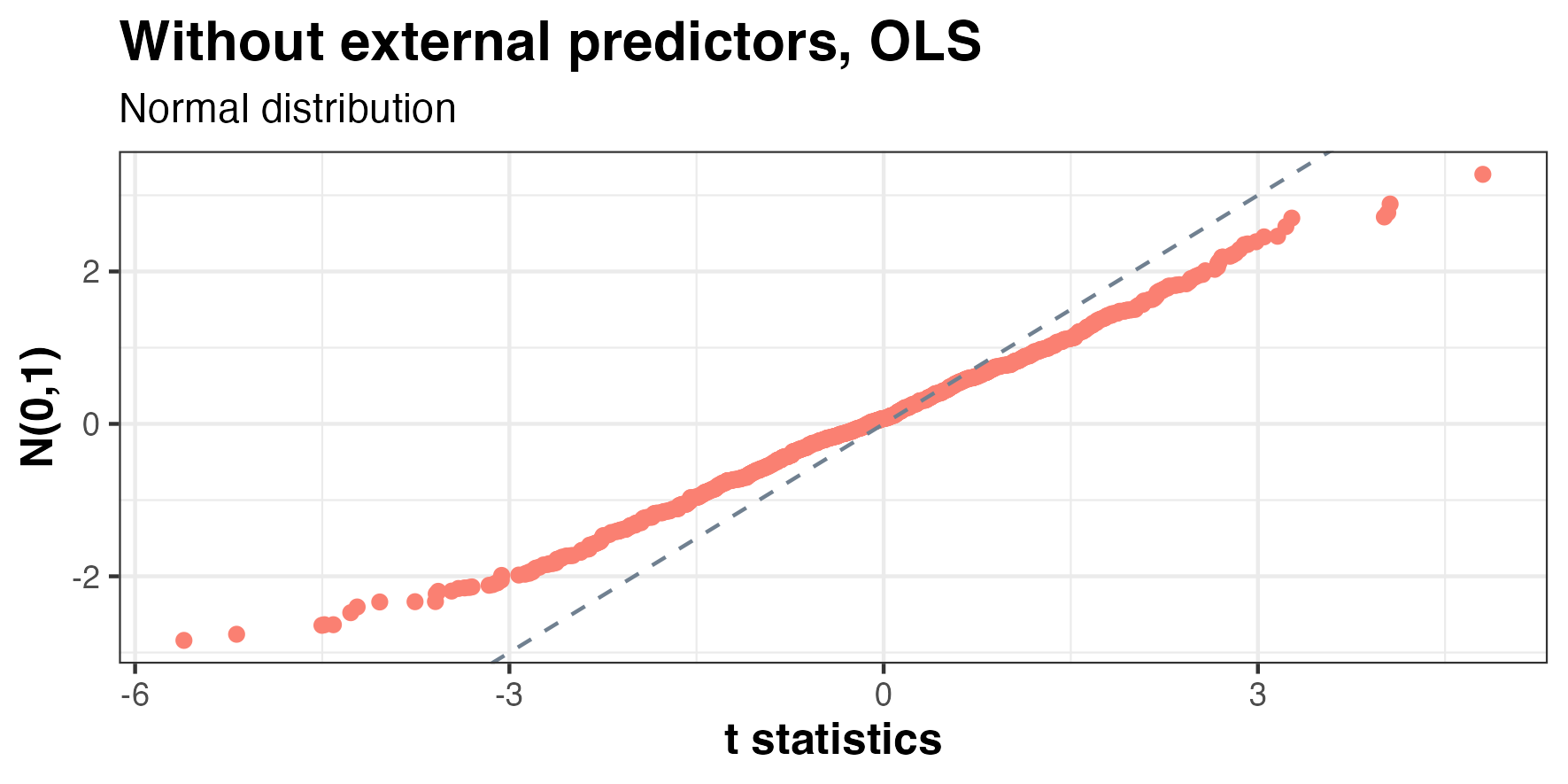}
         \caption{Normal Distribution}
     \end{subfigure}
     \begin{subfigure}[b]{0.48\textwidth}
         \centering
         \includegraphics[width=\textwidth]{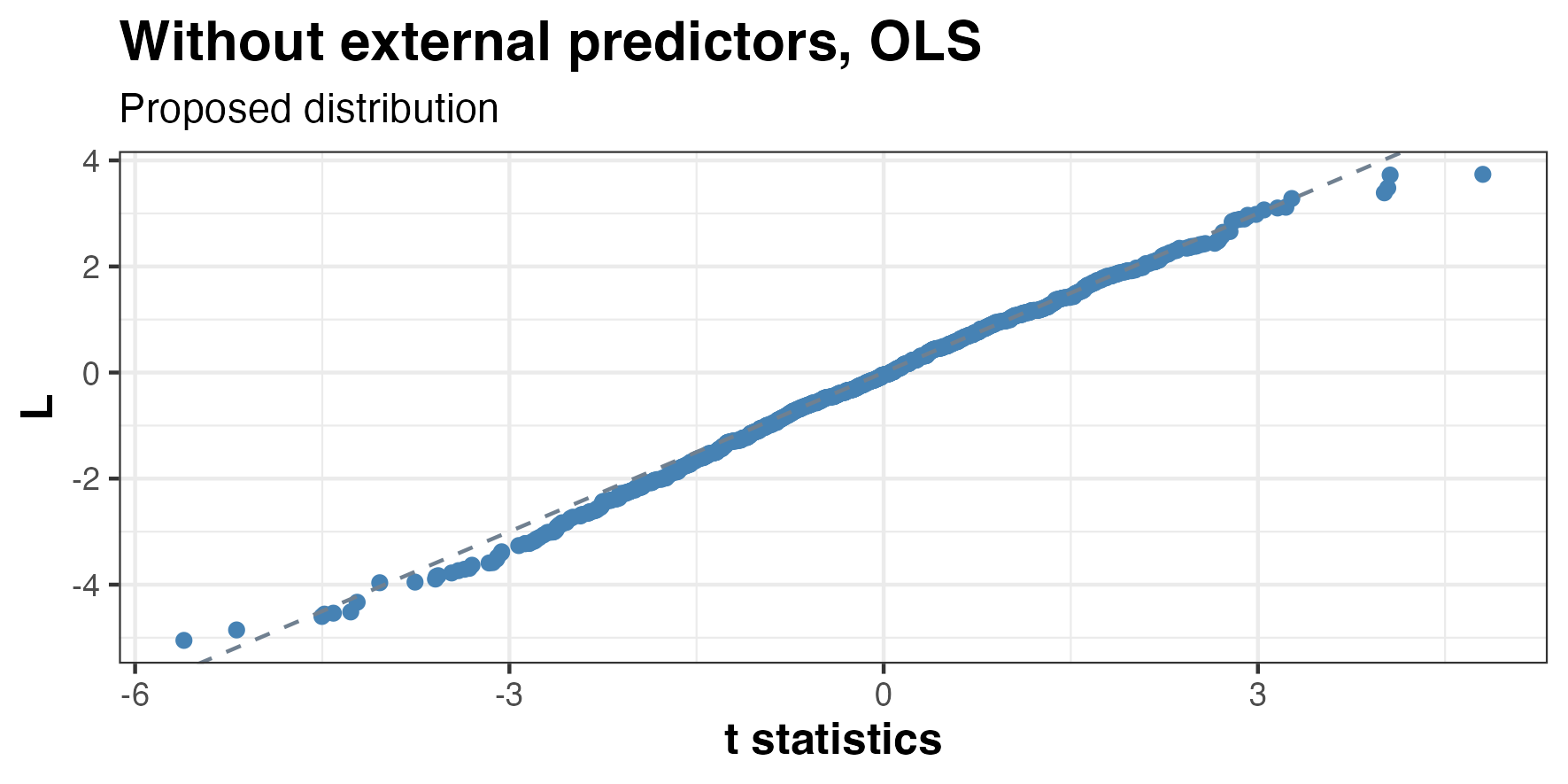}
         \caption{Proposed Distribution}
     \end{subfigure}
     \hfill
     \caption{\em{Comparison with Different Distributions under Model without External Predictors}}
    \label{fig:lsnoX}
\end{figure}

Following the basic framework of the pre-validation procedure, this paper aims to (1) loosen the independence assumption between internal and external predictors, which can lead to invalidity of the permutation test, and instead propose an analytical asymptotic distribution to account for correlation; (2) investigate other common learning procedures that allow for coefficient shrinkage (e.g. ridge regression), and provide analytical asymptotic distributions; (3) for models that do not have closed form solutions for the leave-one-out fits (e.g. Lasso regression), provide a general parametric bootstrapping approach that well approximates the asymptotic null distribution; (4) investigate error properties, i.e. quality of the error estimates, of the pre-validation procedure.

We start with an introduction to the pre-validation procedure in \cref{sec:recap}, and give a recap of the distribution of the test statistic for pre-validated predictors. In \cref{sec:analytic}, we propose an analytical asymptotic distribution of the test statistics when least squares or ridge regression is applied to the first stage of pre-validation. We also provide a brief example of an application to breast cancer study. In \cref{sec:bootstrap}, we provide a general bootstrapping approach to deal with the situations where the asymptotic distribution is not tractable analytically and permutation test is not valid either. We also discuss a handful of computational tricks to accelerate the bootstrap procedure. In \cref{sec:Simulation}, we display simulation results for both the analytical distribution and the bootstrap distribution. In \cref{sec:GWAS}, we further illustrate the merits of pre-validation procedure and bootstrap-based hypothesis testing, by applying them to a synthetic genome-wide association study (GWAS). We also propose a novel approach powered by pre-validation to estimate heritability, i.e. variance of phenotype that can be explained by genotypes. In \cref{sec:error}, we investigate error properties of the pre-validation procedure and provide empirical evidence of good error estimates. We conclude with some discussion in \cref{sec:discussion}.

\section{Pre-validation} \label{sec:recap}
\subsection{Methodology}
The pre-validation procedure is essentially a combination of the leave-one-out method and two-stage regression. Let $n$ be sample size, $y \in \mathbb{R}^n$ be vector of responses (e.g. whether cancer is benign), $Z \in \mathbb{R}^{n \times p}$ be matrix of internal variables (e.g. genetic expressions), and $X \in \mathbb{R}^{n \times e}$ be matrix of external variables (e.g. clinical features), usually $p \gg n$. The procedure can be described as follows:
\begin{enumerate}
    \item At the first stage, leave out the $i^{th}$ observation each time, $i=1,2,...,n$, fit the regression model on the remaining $n-1$ samples $\{(y_k, z_k)\}_{k\neq i}$, and make a prediction of the $i^{th}$ response $y_{PV,i}$ based on $z_i$.
    \item At the second stage, fit a regression model using external predictors and pre-validated predictors $\{\left(y_i, x_i, y_{PV,i}\right)\}_{i=1}^n$, and conduct inference.
\end{enumerate}

\cref{alg:LOO_PV} shows the above procedure in pseudo-code:
\begin{algorithm}
\caption{Leave-one-out Pre-validation Procedure}\label{alg:LOO_PV}
\begin{algorithmic}
\For{$i$ in $1:n$}
\State $f_{(i)}\gets$ regression model $y_{(i)}\sim Z_{(i)}$ \Comment{Leave out $i^{th}$ observation}
\State $y_{PV,i}\gets f_{(i)}\left(z_i\right)$ \Comment{Predict $i^{th}$ response}
\EndFor
\State
\State $g\gets$ regression model $y \sim \left(X, y_{PV}\right)$ \Comment{Use pre-validated estimates at $2^{nd}$ stage}
\State $t\gets \frac{\hat{\beta}_{PV}}{\hat{\sigma}\left(\hat{\beta}_{PV}\right)}$ from model $g$ \Comment{Test hypothesis}
\end{algorithmic}
\end{algorithm}

Note that using the other $n-1$ samples to train model and make prediction at the first stage is essential in pre-validation: since we  fit the second-stage model on those pre-validated predictions, including $y_i$ in training set for the $i^{th}$ sample would be considered as data-reuse and essentially ``cheating". \citeasnoun{tibshirani2002pre} have shown that re-using prediction fit on the full dataset can lead to serious overfitting problem.

Moreover, there are extra benefits of using leave-one-out regression instead of general cross-validation, which divides the training set into $k$-folds (same as leave-one-out if $k=n$). It is a more efficient use of data and can avoid potential issues caused by insufficient sample size when training the first-stage model.

\subsection{Related Work}
As mentioned above, \citeasnoun{hofling2008study} proposed an asymptotic distribution under null hypothesis in a two-stage linear regression procedure, when internal and external predictors are assumed to be independent.

To be more specific, they considered two cases. In the first case, they assume there are no external variables $X$. Let $H = Z(Z^TZ)^{-1}Z^T$ be the projection matrix of leave-one-out linear regression, $D = \text{diag}\{(H_{ii})_{i=1}^n\}$ be the matrix with diagonal elements of $H$, and $I$ be the $n\times n$ identity matrix. Then the pre-validation procedure can be expressed as:
\begin{equation*}
\begin{cases}
    &y_{PV}= (I-D)^{-1}(H-D)y \\
    &y= y_{PV}\beta_{PV} + \epsilon
\end{cases}
\end{equation*}
where $\epsilon \sim \mathcal{N}(0, \sigma^2I)$, and the hypothesis of interest is $H_0:\beta_{PV} = 0$. They have shown that a larger set of internal variables, i.e. a larger $p$, results in a greater departure of asymptotic distribution from normal for the t-statistic. In particular, as $n \rightarrow \infty$,
\begin{equation}\label{t_noX}
    t=\frac{\hat{\beta}_{PV}}{\hat{\sigma}\left(\hat{\beta}_{PV}\right)} \xrightarrow{d} \frac{C-p}{\sqrt{C}}, \quad \quad C\sim \chi_p^2
\end{equation}

The second case is an extension to the above, where they allow for external variables $X$ in the second stage of pre-validation. They assume the distribution of $X$ only depends on $y$: $X_{ik} = y_i + \xi_{ik}$, $i=1,...,n$, $k=1,...,e$, where $\xi_{ik} \overset{\text{i.i.d.}}{\sim} N(0,\sigma^2_k)$. Similar to the above, now the procedure becomes:

\begin{equation*}
\begin{cases}
    &y_{PV}= (I-D)^{-1}(H-D)y \\
    &y= y_{PV}\beta_{PV} + X\beta +\epsilon
\end{cases}
\end{equation*}
where $\beta \in \mathbb{R}^e$ is the regression coefficient vector for external variables. They have shown that under the same null hypothesis $H_0:\beta_{PV} = 0$, t-statistic now have asymptotic distribution:
\begin{equation}\label{t_withX}
    t=\frac{\hat{\beta}_{PV}}{\hat{\sigma}\left(\hat{\beta}_{PV}\right)} \xrightarrow{d} \frac{N^TN-p}{\sqrt{N^TN}} - \frac{N^TA(\mathbf{11^T}+Cov(\xi))^{-1}\mathbf{1}}{\sqrt{N^TN} (1-\mathbf{1^T(\mathbf{11^T}+Cov(\xi))^{-1})}\mathbf{1})}
\end{equation}
where $N \sim \mathcal{N}(0,I_p)$, $A= (A_1,...,A_e)$ with $A_k \sim \mathcal{N}(0,\sigma^2_k I_p)$, $\mathbf{1}= (1,...,1)^T \in \mathbb{R}^e$, and $Cov(\xi)= \text{diag}(\sigma^2_1,..., \sigma^2_e)$.

We treat the above analysis as a benchmark for our analytical results in the next section. We first propose a new model formulation that allows for dependence between the internal and external variables, and then continue to discuss models beyond least squares regression.

\section{Analytical Distribution of the Test Statistic} \label{sec:analytic}
We begin with loosening the independence assumption of internal and external variables. Since our goal is to conduct hypothesis testing on the coefficient of pre-validated predictor, our primary interest is the following null model:
\begin{equation}
    \begin{aligned}
        y&= X\beta_0 + \epsilon\\
        Z&= X\Gamma + E,
    \end{aligned}
\end{equation}
where responses $y \in \mathbb{R}^n$, external variables $X \in \mathbb{R}^{n \times e}$, internal variables $Z \in \mathbb{R}^{n \times p}$, correlation between them is captured by $\Gamma \in \mathbb{R}^{e \times p}$, entries of $E \overset{\text{i.i.d.}}{\sim} N(0,\sigma^2_z)$, and entries
of $\epsilon \overset{\text{i.i.d.}}{\sim} N(0,\sigma^2_x)$.

Throughout, we assume the second-stage regression of $y$ on external predictor $X$ and pre-validated predictor $y_{PV}$ incorporates an intercept term. Therefore, $X$ should be augmented by adding a column of all 1’s. Now, the null model can be interpreted as: conditioning on external variables $X$, the internal variables $Z$ are not predictive of $y$.

\subsection{Least Squares Regression in the First Step} \label{sec:linear}
First, we consider the basic model where least squares regression is applied to both stages of pre-validation.

Define $y_{PV,i}$ to be the predicted value of $y_i$ obtained by regressing $y$ on $Z$ after omitting the $i^{th}$ observation, denoted by $Z_{(i)}$ and $y_{(i)}$. For each $1 \leq i \leq n$, let $z_i$ denote the $i^{th}$ row of $Z$. Now the leave-one-out pre-validated estimate $y_{PV,i}$ can be expressed as:
\begin{equation*}
    y_{PV,i} = z_i(Z^T_{(i)} Z_{(i)})^{-1}Z^T_{(i)}y_{(i)} := (I-D)^{-1}(H-D)y_i
\end{equation*}
where $H = Z(Z^TZ)^{-1}Z^T$ is the projection matrix of leave-one-out linear regression, and $D = \text{diag}\{(H_{ii})_{i=1}^n\}$ is the matrix with diagonal elements of $H$.

Let $y_{PV}$ denote the vector of the $y_{PV,i}$’s, $\beta_{PV}$ be the regression coeﬃcient of $y_{PV}$ when we regress $y$ on $y_{PV}$ and $X$:
\begin{equation*}
    y= y_{PV}\beta_{PV} + X\beta +\epsilon
\end{equation*}

Note that since $X$ already has a column of all 1’s, there is naturally an intercept term in the second-stage regression. Let $\Tilde{X} = \begin{bmatrix}
    y_{PV} \\ X
\end{bmatrix}$, a commonly used estimator of the standard deviation of $\hat{\beta}_{PV}$ is given by the square root of the top-left-most entry of the matrix $\sigma_x^2(\Tilde{X}^T \Tilde{X})^{-1}$, which we denote by $\hat{\sigma}\left(\hat{\beta}_{PV}\right)$. We would like to understand the limiting distribution of $t = \frac{\hat{\beta}_{PV}}{\hat{\sigma}\left(\hat{\beta}_{PV}\right)}$ as $n \rightarrow \infty$, under suitable conditions on limiting behavior of $X$.

We propose the following theorem where the signal-to-noise ratio $SNR = \frac{||\beta_0||_2}{\sigma^2_x}$ is fixed, but allowed to be arbitrarily small:
\begin{theorem} \label{thm:linear}
    Let $X'$ be the design matrix, $X$ be the predictor matrix without a column of 1, and $e$ be the number of columns in $X$. Assume that $\sigma^2_x,\sigma^2_z > 0$. Suppose that as $n \rightarrow \infty$, $\sigma^2_x$, $\sigma^2_z$, $e$, $p$, and $\Gamma$ remain fixed, while $\beta_0$ and $X$ vary with $n$ in such a way that \begin{enumerate}
    \item $\sqrt{n}\beta_0$ converges to a vector $\alpha_0 \in \mathbb{R}^e$;
    \item $\frac{1}{n}X^TX$ converges to an positive definite matrix $\Sigma_{e \times e}$;
    \item $\frac{1}{n}\mathbf{1}^TX$ converges to a vector $\Theta_{1 \times e }$, where 1 denotes the vector of all 1’s;
    \item $\max_{1\leq i \leq n}||x_i||= o(n^{1/3})$, where $x_i$ denotes the $i^{th}$ row of $X$ and $||x_i|| $is its Euclidean norm.
\end{enumerate}

Then, as $n \rightarrow \infty$,
\begin{equation*}
t=\frac{\hat{\beta}_{PV}}{\hat{\sigma}\left(\hat{\beta}_{PV}\right)} \xrightarrow[n \rightarrow \infty]{d} L,
\end{equation*}
where $L$ is a random variable defined as follows. Let two matrices:
\begin{equation*}
    B := \begin{bmatrix}
        1&&\Theta\Gamma\\
        \Gamma^T\Theta^T &&\Gamma^T\Sigma\Gamma +\sigma^2_zI
    \end{bmatrix},\quad
    D:= \begin{bmatrix}
        \Theta\\
        \Gamma^T\Sigma
    \end{bmatrix}
\end{equation*}.

Let $(P_0,Q_0)$ be jointly normal random vectors of dimensions $p+1$ and $e$ respectively, with mean zero and joint covariance matrix $\sigma^2_x \begin{bmatrix}
B&&D\\
D^T&&\Sigma \end{bmatrix}$.

Then, \begin{equation} \label{eq:linear}
    L=\frac{(P_0 + D\alpha_0)^T B^{-1}(P_0-D\Sigma^{-1}Q_0)-\sigma_x^2 (p+1)}{\sigma_x\sqrt{(P_0+ D\alpha_0)^T (B^{-1}-B^{-1}D \Sigma^{-1}D^T B^{-1})(P_0 + D\alpha_0)}},
\end{equation}
where the term within square root in the denominator is strictly positive with probability one.
\end{theorem}
In practice, one should simply take $\alpha_0 = \sqrt{n}\beta_0$, $\Sigma = \frac{1}{n}X^TX$, $\Theta=\frac{1}{n}\mathbf{1}^TX$ if the parameters of the problem are known, or use plug-in estimators if they are unknown. \autoref{fig:lsX} plots the quantiles of standard Normal and our proposed distribution against the quantiles of test statistics from simulation, under model with external predictors. It shows that our proposed distribution gives a better fit to the distribution of test statistics, as the quantile-quantile plot matches the $45^{\circ}$ line better. More simulation results with different configurations of SNR and correlation between $X$ and $Z$ are presented in \cref{sec:Simulation}.

\begin{figure}[H]
     \centering
     \begin{subfigure}[b]{0.48\textwidth}
         \centering
         \includegraphics[width=\textwidth]{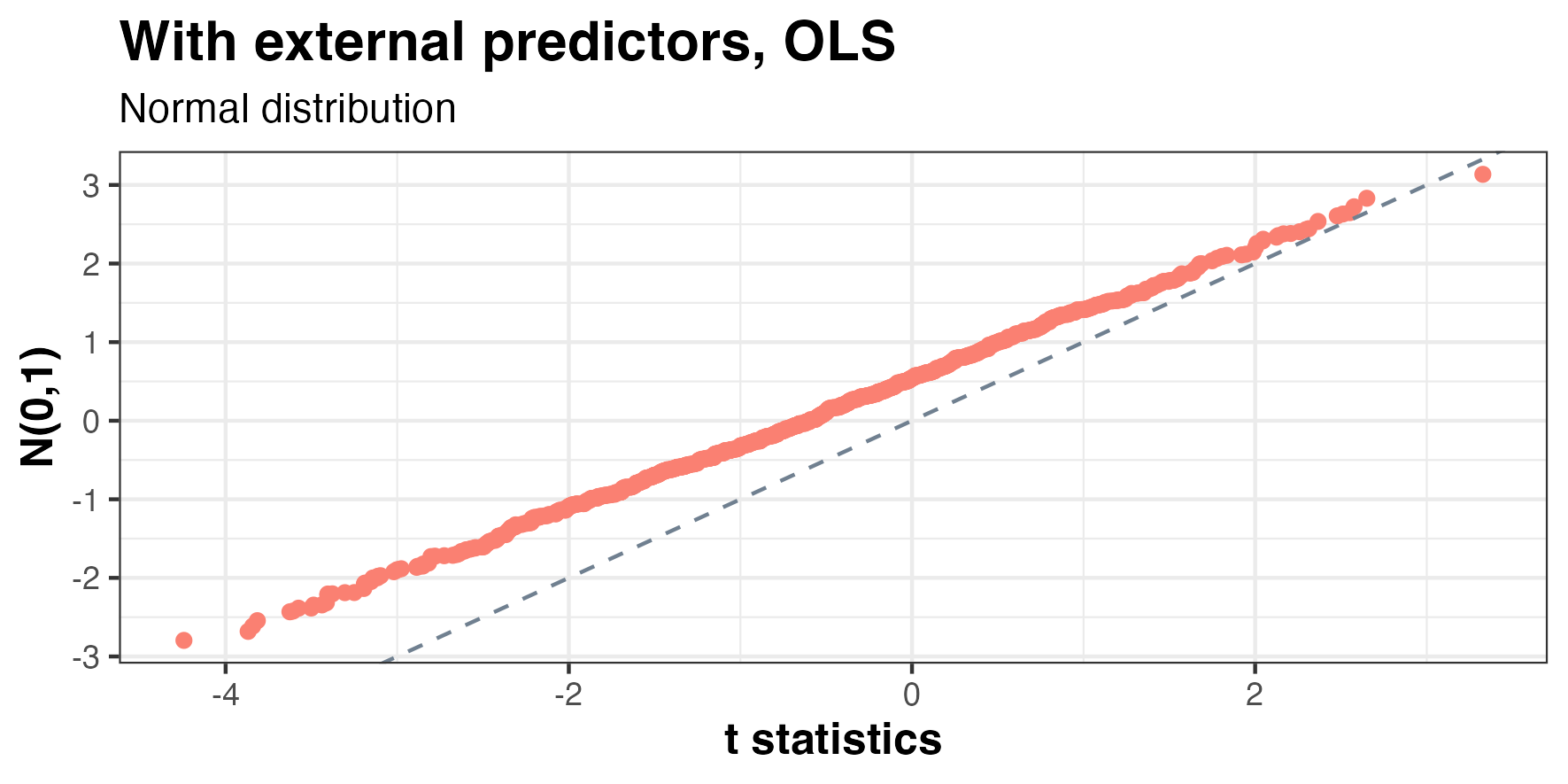}
         \caption{Normal Distribution}
     \end{subfigure}
     \begin{subfigure}[b]{0.48\textwidth}
         \centering
         \includegraphics[width=\textwidth]{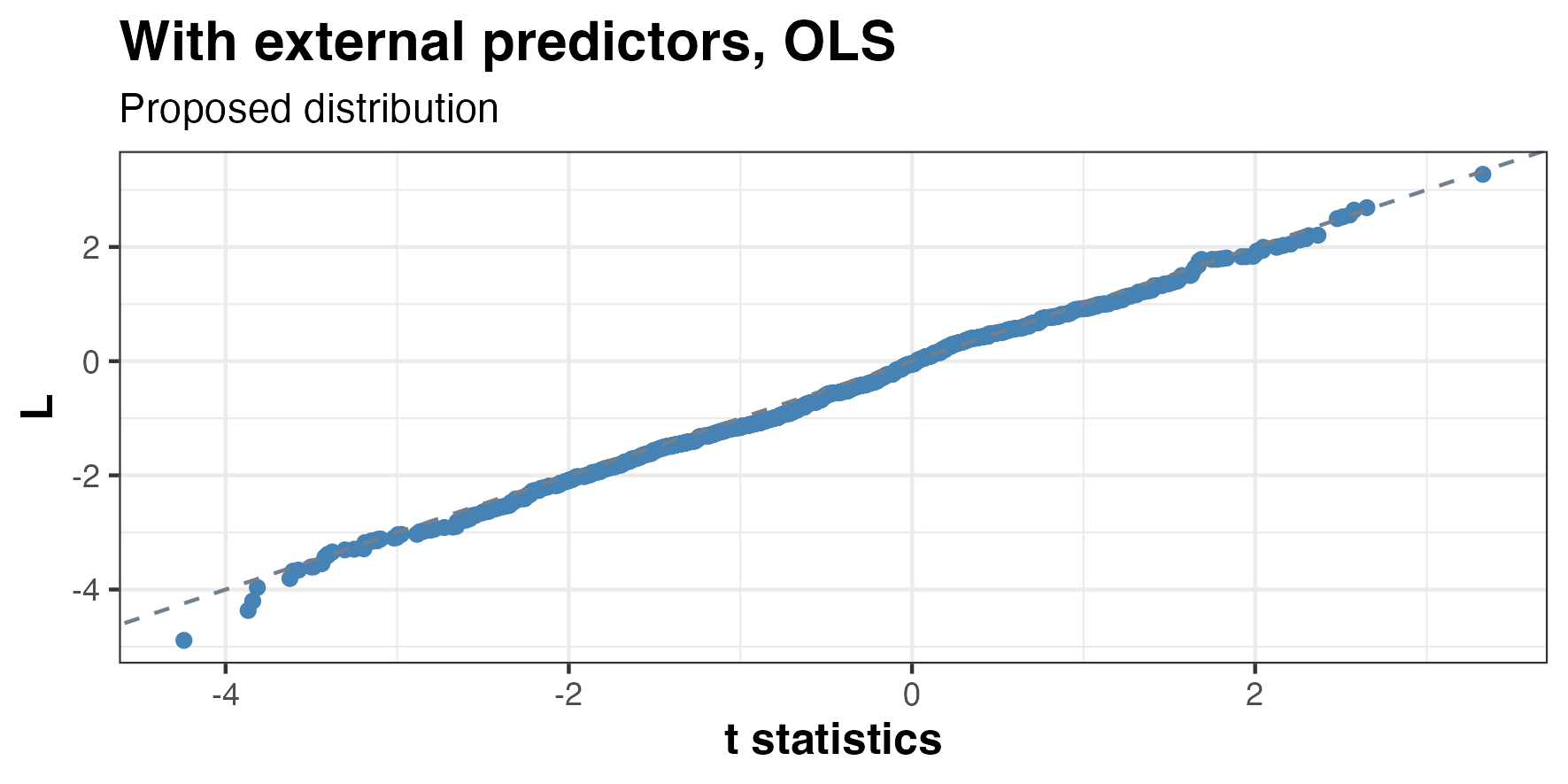}
         \caption{Proposed Distribution}
     \end{subfigure}
     \hfill
     \caption{\em{Comparison with Different Distributions under Model with External Predictors}}
    \label{fig:lsX}
\end{figure}

Comparing \eqref{eq:linear} of our analysis and the results from \citeasnoun{hofling2008study}, we have the following corollary, with proof attached in \cref{appendix:equivalence}. Note that we did not directly compare \eqref{eq:linear} and \eqref{t_withX}, i.e. when there are external predictors, because the generating model is modified to account for correlation between internal and external predictors.

\begin{corollary} \label{cor:equivalence}
    If $\Gamma = \mathbf{0}_{e \times p}$, $\Sigma = I_{e\times e}$, $\Theta = \mathbf{0}_{1 \times e}$, $\sigma^2_1 = ... \sigma^2_e = \sigma^2_x = 1$, $\beta_0 = \alpha_0 = \mathbf{0}_{e \times 1}$, then our proposed distribution will be reduced to \eqref{t_noX}.
\end{corollary}

\subsection{Ridge Regression in the First Step} \label{sec:ridge}
Now we consider the case when $l_2$-penalty is applied to the first-stage of pre-validation. The motivation of doing so is to achieve coefficent shrinkage on the high-dimensional internal variable set $Z$ and reduce variance. Fortunately, an analytical result can still be derived here, since the closed form solution to ridge regression is well-known \cite**{James2013}. Let $\lambda$ be the penalty parameter of ridge regression, and $\gamma_{(i)}$ be the coefficient of leave-one-out regression, then at the first stage we are solving:
\begin{equation*}
    \min_{\gamma_{(i)} \in \mathbf{R^p}} ||y_{(i)} - Z_{(i)}\gamma_{(i)}||^2_2 + \lambda ||\gamma_{(i)}||^2_2
\end{equation*}
where $Z_{(i)}$ and $y_{(i)}$ are obtained by omitting the $i^{th}$ row of $Z$ and $y$.

Now let $z_i$ denote the $i^{th}$ row of $Z$, then the pre-validated predictor of $y_i$ from ridge  regression can be written as:
\begin{equation*}
    y_{PV,i} = z_i(Z^T_{(i)} Z_{(i)} + \lambda I)^{-1}Z^T_{(i)}y_{(i)}
\end{equation*}

As before, we apply linear regression of $y$ on $X$ and $y_{PV}$ at the second stage of pre-validation. We propose the following asymptotic distribution for the t-statistic of pre-validated predictor:

\begin{theorem}\label{thm:ridge}
    Let $X'$ be the design matrix, $X$ be the predictor matrix without a column of 1, and $e$ be the number of columns in $X$. Assume that $\sigma^2_x,\sigma^2_z > 0$. Suppose that as $n \rightarrow \infty$, $\sigma^2_x$, $\sigma^2_z$, $e$, $p$, and $\Gamma$ remain fixed, while $\lambda$, $\beta_0$ and $X$ vary with $n$ in such a way that \begin{enumerate}
    \item $\frac{\lambda}{n}$ converges to a number $\kappa \geq 0$;
    \item $\sqrt{n}\beta_0$ converges to a vector $\alpha_0 \in \mathbb{R}^e$;
    \item $\frac{1}{n}X^TX$ converges to an positive definite matrix $\Sigma_{e \times e}$;
    \item $\frac{1}{n}\mathbf{1}^TX$ converges to a vector $\Theta_{1 \times e }$, where 1 denotes the vector of all 1’s;
    \item $\max_{1\leq i \leq n}||x_i||= o(n^{1/3})$, where $x_i$ denotes the $i^{th}$ row of $X$ and $||x_i|| $is its Euclidean norm.
\end{enumerate}

Then, as $n \rightarrow \infty$,
\begin{equation*}
t=\frac{\hat{\beta}_{PV}}{\hat{\sigma}\left(\hat{\beta}_{PV}\right)} \xrightarrow[n \rightarrow \infty]{d} L,
\end{equation*}
where $L$ is a random variable having the following description. Define two matrices:
\begin{equation*}
    B := \begin{bmatrix}
        1+\kappa&&\Theta\Gamma\\
        \Gamma^T\Theta^T &&\Gamma^T\Sigma\Gamma +(\sigma^2_z+\kappa)I
    \end{bmatrix},\quad
    D:= \begin{bmatrix}
        \Theta\\
        \Gamma^T\Sigma
    \end{bmatrix}
\end{equation*}.

Let $(P_0,Q_0)$ be jointly normal random vectors of dimensions $p+1$ and $e$ respectively, with mean zero and joint covariance matrix $\sigma^2_x \begin{bmatrix}
B&&D\\
D^T&&\Sigma \end{bmatrix}$.

Then, \begin{equation} \label{eq:ridge}
    L=\frac{(P_0 + D\alpha_0)^T B^{-1}(P_0-D\Sigma^{-1}Q_0)-\sigma_x^2 (p+1) + \sigma^2_x\kappa \text{Tr}(B^{-1})}{\sigma_x\sqrt{(P_0+ D\alpha_0)^T (B^{-1}-B^{-1}D \Sigma^{-1}D^T B^{-1})(P_0 + D\alpha_0)}},
\end{equation}
where the term within square root in the denominator is strictly positive with probability one.
\end{theorem}

Again, in practice one should simply take $\alpha_0 = \sqrt{n}\beta_0$, $\kappa = \frac{\lambda}{n}$, $\Sigma = \frac{1}{n}X^TX$, $\Theta=\frac{1}{n}\mathbf{1}^TX$ if the parameters of the problem are known, or use plug-in estimators if they are unknown. Proof of \cref{thm:ridge} can be found in \cref{appendix:ridge}. Note that \cref{thm:linear} is just a special case where $\lambda = 0$.

As a sanity check, one can verify that \eqref{eq:ridge} reduces to \eqref{eq:linear} when $\kappa = 0$. This can happen in two cases: first, if we do not add any penalty term, then ridge regression reduces to least squares by definition; second, if $\lambda$ is fixed and we let $n \rightarrow \infty$, then the ridge estimates will converge to least squares estimates, as OLS estimates are asymptotically optimal.

\section{Bootstrap Distribution of the Test Statistic} \label{sec:bootstrap}
In the previous section, we derived analytical distributions for the t-statistic when least squares or ridge regression is applied to the first stage of pre-validation. Our derivation is based on the analytical form of the projection matrix for leave-one-out estimates. As one may expect, the nice closed-form solutions are not available in many other settings, such as Lasso regression. Therefore, we  provide a general approach that can accommodate variable selection methods applied to first-stage regression, especially when the asymptotic distribution is not tractable analytically.

We first give details of  a bootstrap method \cite{efron1993introduction} for approximating the distribution of the test statistic $t = \frac{\hat{\beta}_{PV}}{\hat{\sigma}\left(\hat{\beta}_{PV}\right)}$, where $\hat{\beta}_{PV}$ is regression coefficient of the pre-validated predictor, and then introduce some computational tricks to accelerate the process. We  show the performance of this approach on simulated data in \cref{sec:Simulation}, and apply it to GWAS analysis in \cref{sec:GWAS}.

\subsection{Bootstrap Procedure} \label{sec:nonparaboot}
First, we define notations and parameters used in the bootstrap procedure: let $n$ be sample size of the observed dataset $D = (y, X, Z)$, $B$ be the number of random bootstrap samples, and $n_{boot}$ be the sample size of each bootstrap sample $\left\{ \left(y^{*b}, X^{*b}, Z^{*b}\right)\right\}_{b=1}^B$, where each sample is obtained by resampling from the original dataset $D$.

Now, the bootstrap procedure for approximating distribution of $t$ is as follows: first, resample $n$ observations from observed dataset $D$ with replacement, and construct a vector of weights $\mathbf{w}^{*1} =(w_1^{*1},...w_n^{*1})$ counting the occurrences of each observation. For example, if $i^{th}$ observation is selected twice, then $w_i^{*1} = \frac{2}{n}$; if $j^{th}$ observation is never selected, then $w_j^{*1} = 0$. Note that for this particular bootstrap replication, only those observations with non-zero weights will be used, hence we can simply discard the observations with zero weights and use the rest as the first bootstrap sample of size $n_{boot}$. Then, run the two-stage leave-one-out pre-validation on $\left(y^{*1}, X^{*1}, Z^{*1}\right)$, with Lasso applied to the first stage to perform variable selection. At the end of second stage, we  obtain the first bootstrap replication of t-statistic, denoted by $t^{*1}$. Repeat the procedure for $B$ times and eventually obtain $\left\{t^{*b}\right\}_{b=1}^B$. Finally, the $B$ bootstrap replications form an approximation  to the distribution of $t$.

\cref{alg:boot} below gives a more explicit description of the above procedure.
\begin{algorithm}
\caption{Bootstrap for Leave-one-out Pre-validation Estimates}\label{alg:boot}
\begin{algorithmic}
\For{$b$ in $1:B$}
\State $sample \gets$ $n$ numbers sampled from $1,...n$ with replacement
\State $\mathbf{w}^{*b}\gets$ occurrence counts in $sample$ \Comment{Serve as weights in model fitting}
\State $y^{*b}\gets \{y_i:\, w_i^{*b}\neq 0\}$ \Comment{Discard rows with zero weights}
\State $X^{*b}\gets \{x_i:\, w_i^{*b}\neq 0\}$
\State $Z^{*b}\gets \{z_i:\, w_i^{*b}\neq 0\}$
\State
\For{$i$ in $1:n_{boot}$}
\State $f_{(i)}\gets$ weighted Lasso regression on $\left(y^{*b}_{(i)}, Z^{*b}_{(i)}\right)$, with weights $\{w_i:\, w_i^{*b}\neq 0\}$ \\ \Comment{Leave out $i^{th}$ observation}
\State $y^{*b}_{PV,i}\gets f_{(i)}\left(z^{*b}_i\right)$ \Comment{Predict $i^{th}$ response}
\EndFor
\State
\State $g^{*b}\gets$ weighted linear regression on $\left(y^{*b}, X^{*b}, y^{*b}_{PV}\right)$ \Comment{Use pre-validated estimates}
\State $t^{*b}\gets \frac{\hat{\beta}^{*b}_{PV}}{\hat{\sigma}\left(\hat{\beta}^{*b}_{PV}\right)}$ from model $g^{*b}$ \Comment{Obtain bootstrap replications}
\EndFor
\end{algorithmic}
\end{algorithm}

\subsection{Parameter Pre-training and the ALO Estimator}
In \cref{alg:boot}, notice that we would have to run $n_{boot}$ times of \texttt{cv.glmnet()} within each bootstrap sample. This can be computationally demanding when either of $n$ or $B$ is large. Therefore, we now introduce a more efficient way to compute the leave-one-out pre-validation estimates when the Lasso regression is used in the first stage. The method involves two steps, detailed in the following subsections.

\subsubsection{Pre-training Lasso Penalty Parameter} \label{sec:pretrain}
First, instead of using \texttt{cv.glmnet()} \cite{Friedman2010-rt} for each bootstrap sample, which computes a sequence of Lasso estimates according to different penalty parameters $\lambda$ in a parallel manner, pre-training the penalty parameter with an observed dataset and passing on to \texttt{glmnet()} in each bootstrap sample would accelerate the procedure. In a simulation setting, it can be further reduced to only running \texttt{cv.glmnet()} once with one simulated sample, and pass on to all the following simulations and bootstrap runs.

In addition to the pre-training, we also want to give some recommendations on how to pick the penalty parameter. Rather than the commonly used $lambda.min$ or $lambda.1se$, we would encourage a more conservative choice of $\lambda$. The rationale behind this is that \texttt{cv.glmnet()} would allow for a diverging $\lambda$, especially when the actual reduction of regression MSE brought by the growing $\lambda$ is trivial. In other words, the unbounded growth of $\lambda$ would not benefit the model performance anymore.

Now in the asymptotic regime, this could cause serious issues, since we should expect linear model to have the optimal performance, i.e. $\lambda \rightarrow 0$, as $n \rightarrow \infty$. Therefore, our recommendation would be to choose the smallest possible $\lambda$ such that the increase of regression MSE is within an acceptable margin of the MSE under $lambda.min$. Note that here we define "acceptable" loosely. In fact, it could be 1 standard error, 10\% of the minimal MSE, or any other reasonable rule by choice.

\subsubsection{Approximate Leave-One-Out (ALO) Estimator for Lasso Estimates}\label{sec:approx}
Second, the Lasso estimates for $y^{*b}_{PV}$ from \texttt{glmnet()} would be costly to compute, when either sample size $n$ or number of bootstrap iterations $B$ is large. Instead, we could use the \textit{Approximate Leave-One-Out} (ALO) estimator for Lasso estimates.

\citeasnoun**{auddy2024approximate} proved a theoretical error bound for the ALO estimator using regression with $l_1$ regularizers. One can show that under Gaussian model, the ALO estimator has the following form:
\begin{equation}\label{eq:ALO_est}
\begin{aligned}
    y^{*b}_{PV,i} &\approx  \frac{\hat{y}^{*b}_i - H_{ii}\, y_i}{1 - H_{ii}}\\
     H &= Z_{\mathcal{S}} (Z_{\mathcal{S}}^T Z_{\mathcal{S}})^{-1} Z_{\mathcal{S}}^T
\end{aligned}
\end{equation}
where $y^{*b}_{PV,i}$ is the leave-one-out Lasso estimate, $\hat{y}^{*b}_i$ is the Lasso estimate without leave-one-out cross fitting and $\mathcal{S} = \{i:\hat{\beta}_i\neq 0\}$ is the active set of regular Lasso estimates, both only need to be computed once for each bootstrap sample, and $H$ is the approximate Lasso operator. By using the ALO estimator, we could avoid running $n_{boot}$ leave-one-out regressions for each bootstrap sample. 

We show performance of the ALO estimator by simulation. \autoref{fig:compare} compares the pre-validated Lasso estimates approximated by the ALO estimator and those obtained directly from running $\texttt{glmnet()}$. It shows that the two quantities are on an exact $45^{\circ}$ line, which implies that the ALO estimator produces good approximations.

\begin{figure}[H]
    \centering
    \includegraphics[width=0.5\textwidth]{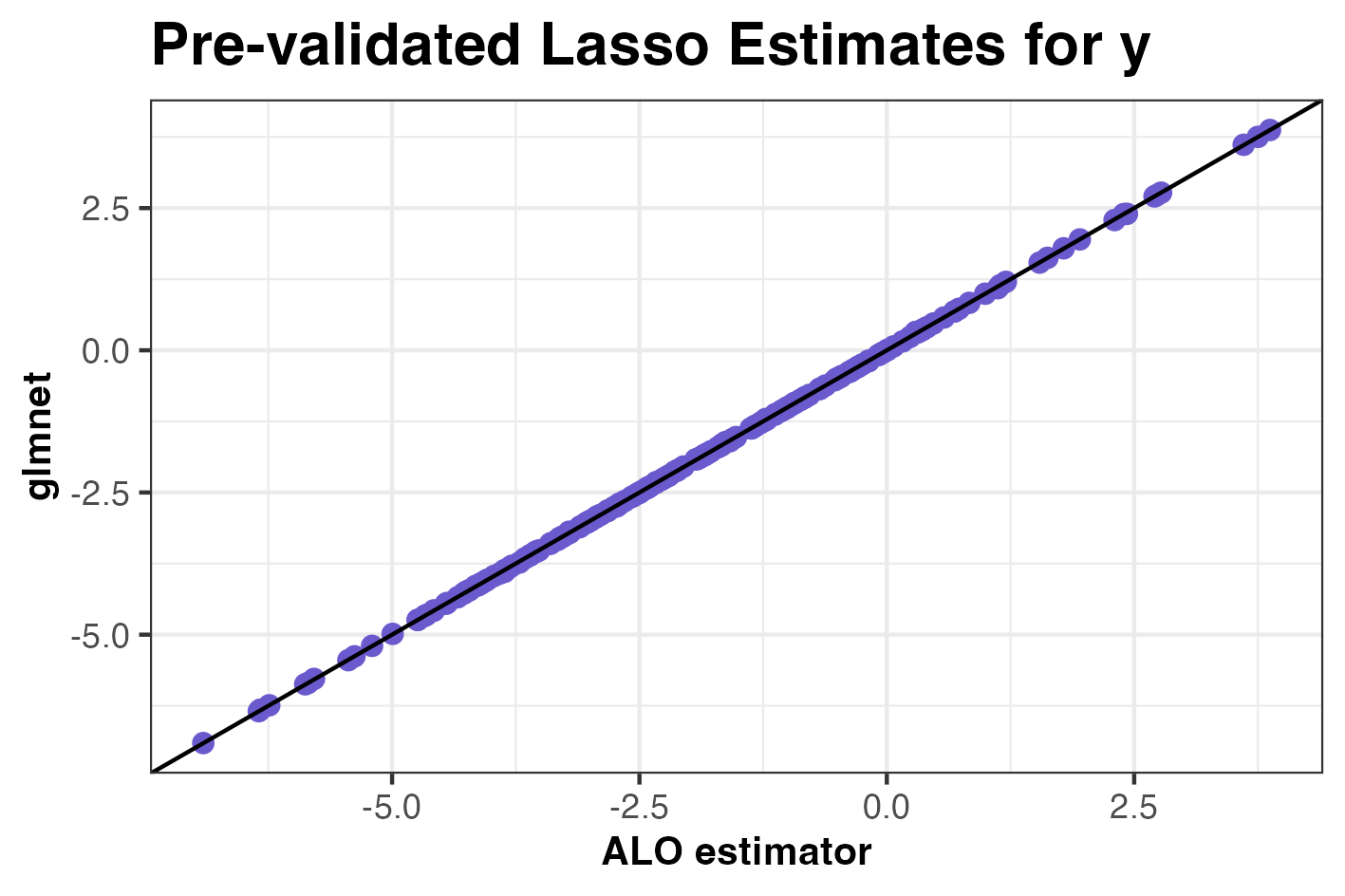}
    \caption{\em{Comparison of Pre-validation Lasso Estimates\\ Obtained by Different Methods}}
    \label{fig:compare}
\end{figure}

\subsection{Hypotheses Testing and the Parametric Bootstrap}
Now that we can apply the bootstrap procedure introduced in \cref{sec:nonparaboot} to obtain an approximate distribution of the test statistic for $H_0: \beta_{PV}=0$. In particular, we would like to investigate the false discovery rate, i.e. probability of rejecting null hypothesis when the null is true, and the power, i.e. probability of rejecting null hypothesis when the alternative is true. We consider the following alternative model: 
\begin{equation} \label{eq:altmodel}
    \begin{aligned}
    y&= X\beta_0 + Z\phi + \epsilon\\
    Z&= X\Gamma + E,
    \end{aligned}
\end{equation}
where $\phi \in \mathbb{R}^p$, $\phi \neq \mathbf{0}$. The alternative model implies that conditioning on external variables $X$, the internal variables $Z$ are still predictive of response $y$. Notice that we did not impose any further restriction on the structure of $(X,Z)$ here, for instance independency.

One might have the concern whether strong multi-collinearity would make the test invalid. Here, we show heuristically that it would still be an appropriate alternative model, even when $X$ and $Z$ are linearly correlated. We can write \eqref{eq:altmodel} as:
\begin{equation*}
    \begin{aligned}
        y&= X\beta_0 + Z\phi + \epsilon\\
        &= X\beta_0 + (X\Gamma + E)\phi + \epsilon \\
        &= X(\beta_0+\Gamma\phi)+E\phi+ \epsilon \\
        &:= X\Tilde{\beta_0}+\Tilde{Z}\phi+ \epsilon 
    \end{aligned}
\end{equation*}
then we have the newly defined $\Tilde{Z} = E$ independent of $X$. This shows that the test always focuses on the information of $Z$ that $X$ does not share. In fact, the structure of $(X,Z)$ would not impact the procedure.

However, another main concern here is whether conducting hypothesis testing based on the bootstrap distribution would have any power. Intuitively, since the bootstrap procedure aims to replicate the true alternative distribution when the data come from the alternative, then power of the test could be a serious issue. \autoref{fig:naive_boot} confirms that this is a valid concern: the bootstrap distribution generated from an alternative dataset is very close to the true alternative distribution. This could lead to bad power, since it would be hard to reject any hypothesis when data is from a true alternative.

\begin{figure}[H]
    \centering
    \includegraphics[width=1\textwidth]{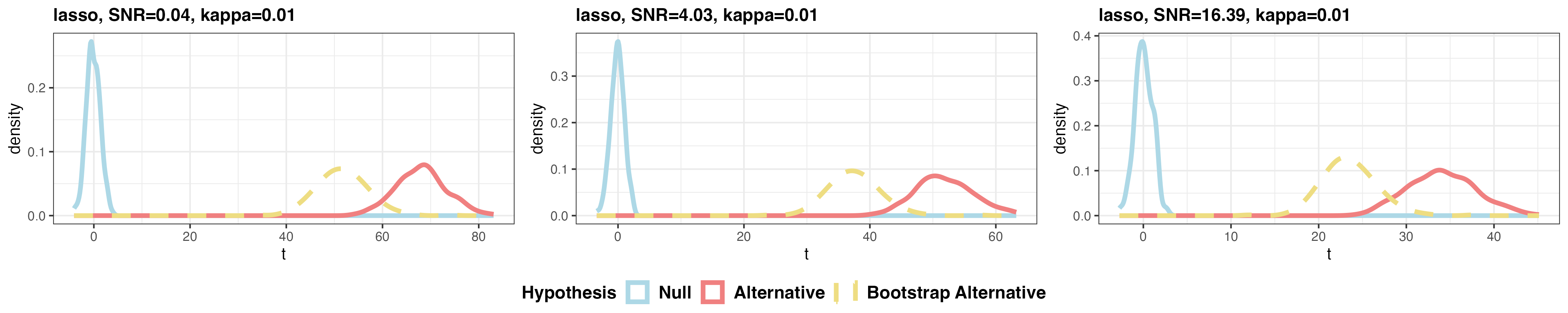}
    \caption{\em{Comparison of True Null, True Alternative, and Bootstrap Distribution Generated from Alternative Data:} \small the bootstrap distribution generated from an alternative dataset is very close to the true alternative distribution, which could lead to bad power.}
    \label{fig:naive_boot}
\end{figure}

To fix this problem, we propose a parametric bootstrap --- a variation of \cref{alg:boot} that can generate a bootstrap distribution replicating the true null, when the data comes from an alternative. The procedure is as follows: first apply linear regression of $y$ on $X$, obtain the regression coefficient $\hat{\beta}$ and compute the residual vector $\xi = y-X\hat{\beta}$; then sample from $\xi$ with replacement to obtain a bootstrap sample $y^{*b} = X\hat{\beta} + \xi^{*b}$, while keep $(X,Z)$ fixed; lastly proceed the pre-validation procedure with $(y^{*b}, X, Z)$. A detailed pseudo-code can be found in \cref{alg:paraboot}.

\begin{algorithm}
\caption{Parametric Bootstrap for Leave-one-out Pre-validation Estimates}\label{alg:paraboot}
\begin{algorithmic}
\For{$b$ in $1:B$}
\State $\hat{\beta}\gets$ regression coefficient of $y$ on $X$ \Comment{Using observed dataset}
\State $\xi \gets y-X\hat{\beta}$ \Comment{Regression residuals}
\State $\xi^{*b}_1, ...\xi^{*b}_n \gets$ sampling $n$ observations from $\xi$ with replacement
\State $y^{*b}_i\gets x_i^T\hat{\beta} + \xi^{*b}_i$
\Comment{Keep $(X,Z)$ fixed}
\State
\For{$i$ in $1:n_{boot}$}
\State $f_{(i)}\gets$ Lasso penalized linear regression on $\left(y^{*b}_{(i)}, Z_{(i)}\right)$ \Comment{Leave out $i^{th}$ observation}
\State $y^{*b}_{PV,i}\gets f_{(i)}\left(z_i\right)$ \Comment{Predict $i^{th}$ response}
\EndFor
\State
\State $g^{*b}\gets$ linear regression on $\left(y^{*b}, X, y^{*b}_{PV}\right)$ \Comment{Use pre-validated estimates at $2^{nd}$ stage}
\State $t^{*b}\gets \frac{\hat{\beta}^{*b}_{PV}}{\hat{\sigma}\left(\hat{\beta}^{*b}_{PV}\right)}$ from model $g^{*b}$ \Comment{Obtain bootstrap replications}
\EndFor
\end{algorithmic}
\end{algorithm}

Note that we no longer need to construct a weight vector based on occurrence counts in parametric bootstrap. \cref{alg:paraboot} guarantees type-I error control, i.e. controls false discovery rate, because it still replicates the null when the data comes from a null distribution.

\cref{alg:paraboot} also provides a valid bootstrap confidence interval, following the theory for bootstrap (see e.g. \citeasnoun{efron1993introduction}). Simulation results computing false discovery rate and power are given in \cref{sec:Simulation}.

\subsection{Example: Breast Cancer Study}

In this section, we apply the proposed bootstrap procedure to inference of microarray features in breast cancer dataset \cite**{breastCancerNKI}. The gene expression data are from a breast cancer study published by \citeasnoun**{van2002gene} and \citeasnoun**{doi:10.1056/NEJMoa021967}. The dataset contains 295 patients with stage I or II breast cancer and were younger than 53 years old; 151 had lymph-node-negaitve disease, and 144 had lymph-node-positive disease.

Both studies used DNA microarray analysis and applied supervised classification to identify a gene expression signature strongly predictive of a short interval to distant metastases (e.g. `poor prognosis’ signature) in patients, without tumour cells in local lymph nodes at diagnosis (e.g. lymph node negative). They showed the gene-expression profile is a more powerful predictor of the disease outcomes in young patients with breast cancer than the currently used clinical parameters. 

Now, we would like to compare two procedures: pre-validation and fitting the first-stage model with all data instead (essentially data-reuse). We compare the z-score, p-value and mis-classification rate on a held-out test set. Furthermore, we would like to compare p-value computed by \texttt{glm()} function with the one from bootstrap procedure. To reduce computational burden, we first pick 1000 out of 24481 gene expression with highest variance to fit the first-stage model. Since the response is binary (positive or negative), we use the logit-transformed fitted probability from first stage as the microarray predictor, then fit a logistic regression at the second stage. We are interested in comparing the microarray predictor to a set of clinical predictors, including tumor size, patient age, estrogen receptor (ER) status and tumor grade. \autoref{fig:modelreuse} and \autoref{fig:modelpv} summarize the results.

\begin{table}[H]
\centering
    \includegraphics[width=0.55\textwidth]{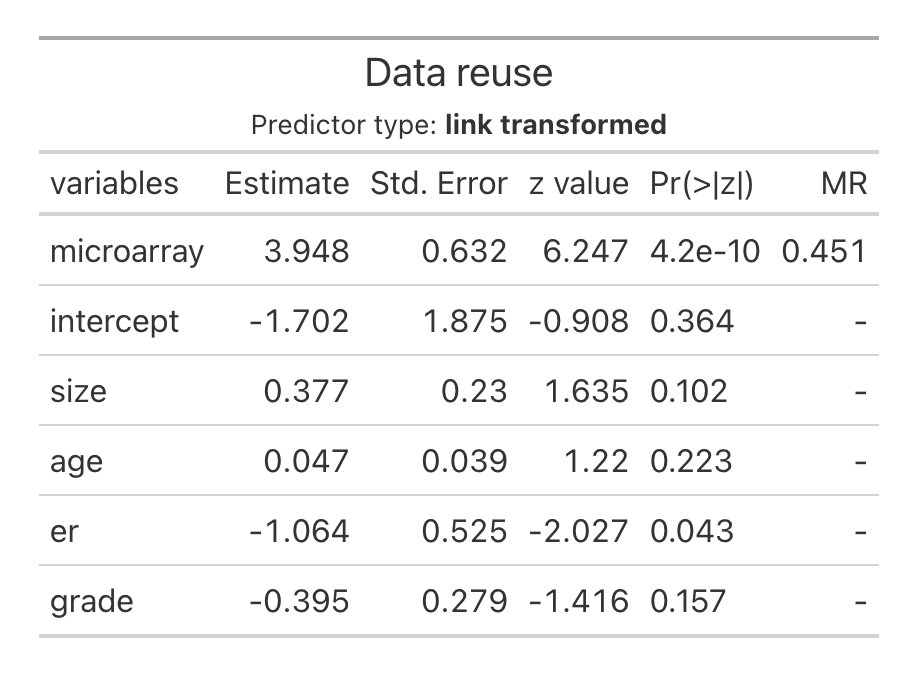}
    \caption{\em{Model Summary for Data Reuse}}
    \label{fig:modelreuse}
\end{table}
\begin{table}[H]
\centering
   \includegraphics[width=1\textwidth]{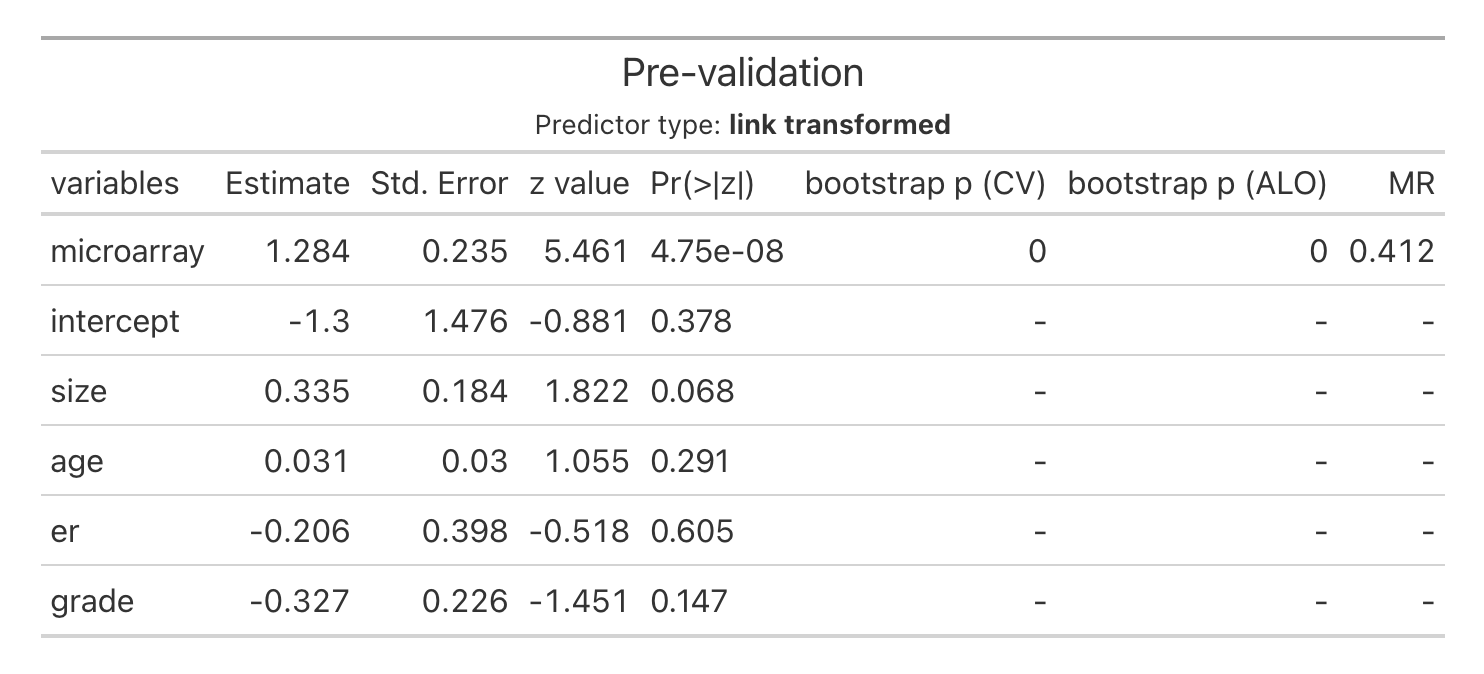}
   \caption{\em{Model Summary for Pre-validation}}
   \label{fig:modelpv}
\end{table}

We can observe that fitting the second-stage logistic regression on the pre-validated microarray predictor gives a lower z-score and higher p-value. In other words, the microarray predictor under pre-validation looks less significant, because it has never seen the true response for each individual in model fitting. Second, we can see that the misclassification rate on a held-out test set is lower under pre-validation, implying it is a better model. Note that we only utilized 1000 gene expression variables in an unsupervised way to fit the model in the first stage, which can lead to a high mis-classification rate in general. Lastly, p-value directly from \texttt{glm()}, p-value computed by \texttt{cv.glmnet()}-based bootstrap procedure, and p-value computed by approximate leave-one-out estimate (ALO)-based bootstrap procedure are all about the same. The derivation of ALO for the logistic Lasso problem can be found in \cref{appendix:ALOlogistic}.

\section{Simulation Results} \label{sec:Simulation}
\subsection{Results for Analytical Distribution}
We first show simulations based on our analytical results in \cref{sec:analytic} for least squares and ridge regression. We also show results using ALO estimator mentioned in \cref{sec:approx} to fast compute the Lasso estimates.

We run simulations with sample size $n = 100$, number of internal variables $p = 30$ and number of external variables $e = 5$ (intercept included), for $n_{sim} = 300$ times. The parameters we tune are: signal-to-noise ratios $\frac{Var(X\beta)}{\sigma^2_x}$, and correlation $\Gamma$ between internal and external variables.

\subsubsection{Least Squares Regression}
\autoref{fig:compare_lin} compares the simulated distribution of the t-statistic with Normal and our analytical distribution respectively. Each row represents a different correlation $\Gamma$ between internal and external variables, while each column represents a different signal-to-noise ratio. For each plot, we keep the same configuration of SNR ($\beta_0$, $\sigma^2_x$ in null model) and $\Gamma$, to compare the proposed distribution with the benchmarks directly.

We can observe the following:

\begin{enumerate}
    \item Even when the simplest least squares regression is applied to the first stage of pre-validation, there is a significant departure of the actual distribution of t-statistic from Normal distribution, especially when signal-to-noise ratios are relatively low, or when correlations between internal and external predictors are relatively high.
    \item Our proposed analytical asymptotic distribution provides a better fit to the actual distribution of t-statistic in terms of bias, variance and tail mis-coverage rate, i.e. type-I error $\alpha$, which is set to be 0.05 as a benchmark. We can see that our proposed distribution has a more precise control of type-I error, as the tail mis-coverage rate is closer to the nominal level.
\end{enumerate}

\begin{figure}[H]
    \centering
    \includegraphics[width=0.95\textwidth]{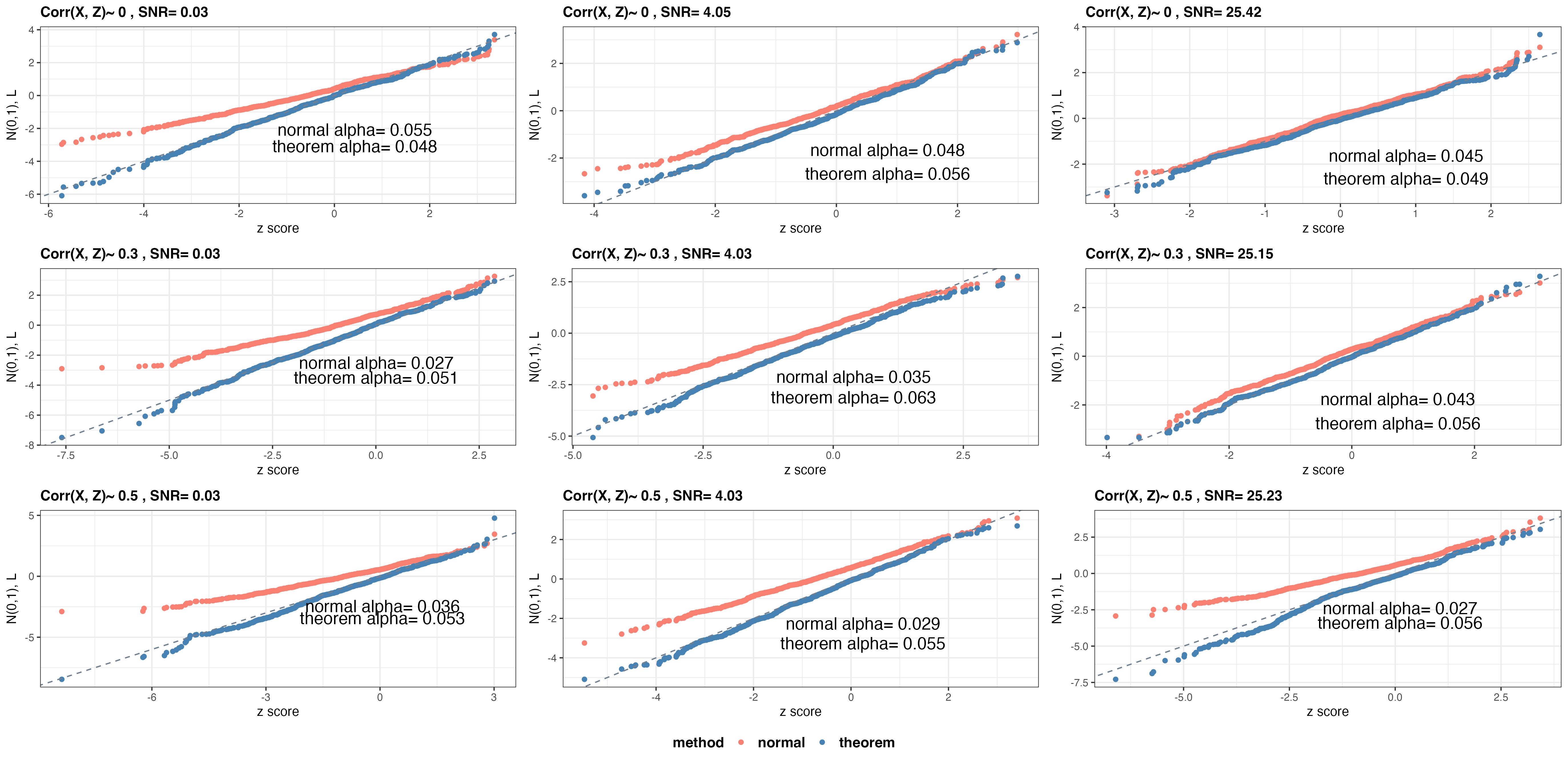}
    \caption{\em{Comparison of Normal and Analytical Distribution using OLS Regression at the first step:}
    \small each row represents a different correlation $\Gamma$ between internal and external predictors, while each column represents a different signal-to-noise ratio. Our proposed analytical distribution outperforms standard Normal especially when the signal-to-noise ratio is low, or when the correlation is high between internal and external predictors. It also has a more precise control of type-I error, as the tail mis-coverage rate is closer to the nominal level.}
    \label{fig:compare_lin}
\end{figure}

\subsubsection{Ridge Regression}
As mentioned above, our analytical distribution requires a fixed $\kappa$ as the input penalty parameter of $\texttt{glmnet()}$, rather than letting $\texttt{cv.glmnet()}$ pick a different $\kappa$ each time. We also recommend pre-training and picking a conservative $\kappa$, as discussed in \cref{sec:pretrain}.

\autoref{fig:compare_ridge} shows comparison of the distribution of t-statistic with Normal and our proposed analytical asymptotic distribution. Again, we can observe an improvement of accuracy in terms of bias, variance, and tail mis-coverage rate, i.e. type-I error $\alpha$, which is set to be 0.05 as a benchmark.
\begin{figure}[H]
    \centering
    \includegraphics[width=0.95\textwidth]{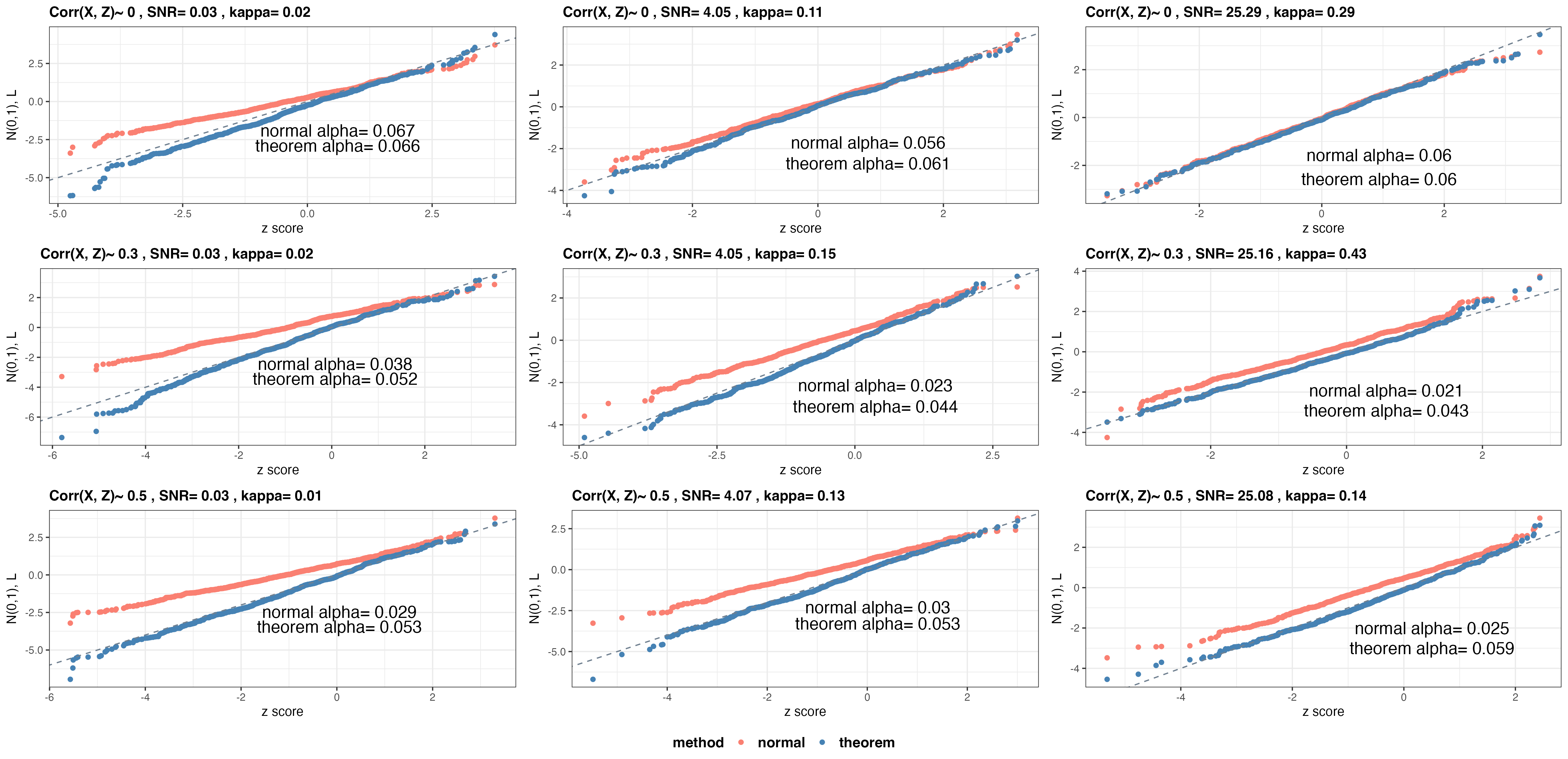}
    \caption{\em{Comparison of Normal and Analytical Distribution using Ridge Regression at the first step:}
    \small Penalty parameter $\kappa$ is chosen via pre-training. As before, our proposed analytical distribution outperforms standard Normal in terms of bias, variance, and tail mis-coverage control, especially when signal-to-noise ratio is low, or when the correlation between two sets of predictors is high.}
    \label{fig:compare_ridge}
\end{figure}

\subsection{Results for Bootstrap Distribution}
In \cref{sec:bootstrap}, we have proposed a general approach to deal with the situation where an analytical distribution of the t-statistic is not tractable. Furthermore, to address the lack of power, we have proposed a modified parametric bootstrap procedure in \cref{alg:paraboot}. \autoref{fig:LOOboot} shows the distributions of true null, true alternative, bootstrap null from null data, and bootstrap null from alternative data, when the  Lasso regression is used in the first stage of pre-validation. 

\begin{figure}[t!]
    \centering
    \includegraphics[width=1\textwidth]{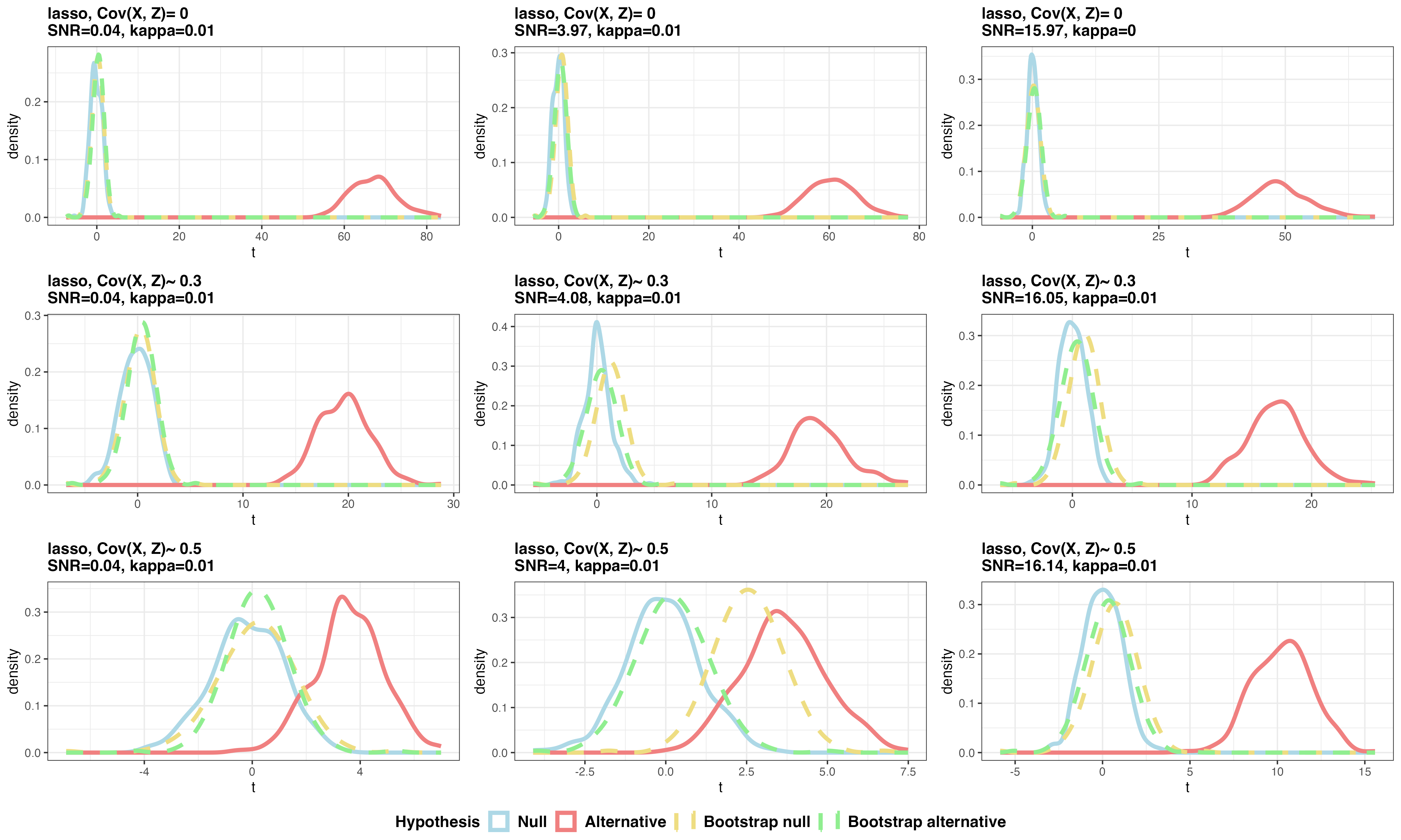}
    \caption{\em{Distributions of True Null, True Alternative, Bootstrap Null from Null Data,\\ and Bootstrap Null from Alternative Data, with Lasso Penalty Applied in First Stage}}
    \label{fig:LOOboot}
\end{figure}

Based on the figure, we can observe that:
\begin{enumerate}
    \item Performance of the bootstrap null approximation is quite satisfactory, for all different configurations of SNR and $\Gamma$, no matter if the data is from the null or alternative distribution.
    \item The bootstrap distributions allow us to conduct statistical inference. For instance, we can estimate a p-value using $\hat{\mathbb{P}}_0\left[|t_{obs}-t^{*b}|>\alpha\right] = \frac{1}{B}\sum_{b=1}^B \mathbf{1}\{|t_{obs}-t^{*b}|>\alpha\}$ while controlling type-I error, where $t^{*b}$ is a bootstrap replication of t-statistic and $\alpha$ is the pre-specified significance level.
\end{enumerate}

\autoref{fig:power} displays type-I error and power of testing the hypothesis $H_0:\beta_{PV} = 0$ under different configurations of SNR and $\Gamma$, where significance level $\alpha = 0.05$:

\begin{table}[H]
    \centering
    \includegraphics[scale=0.25]{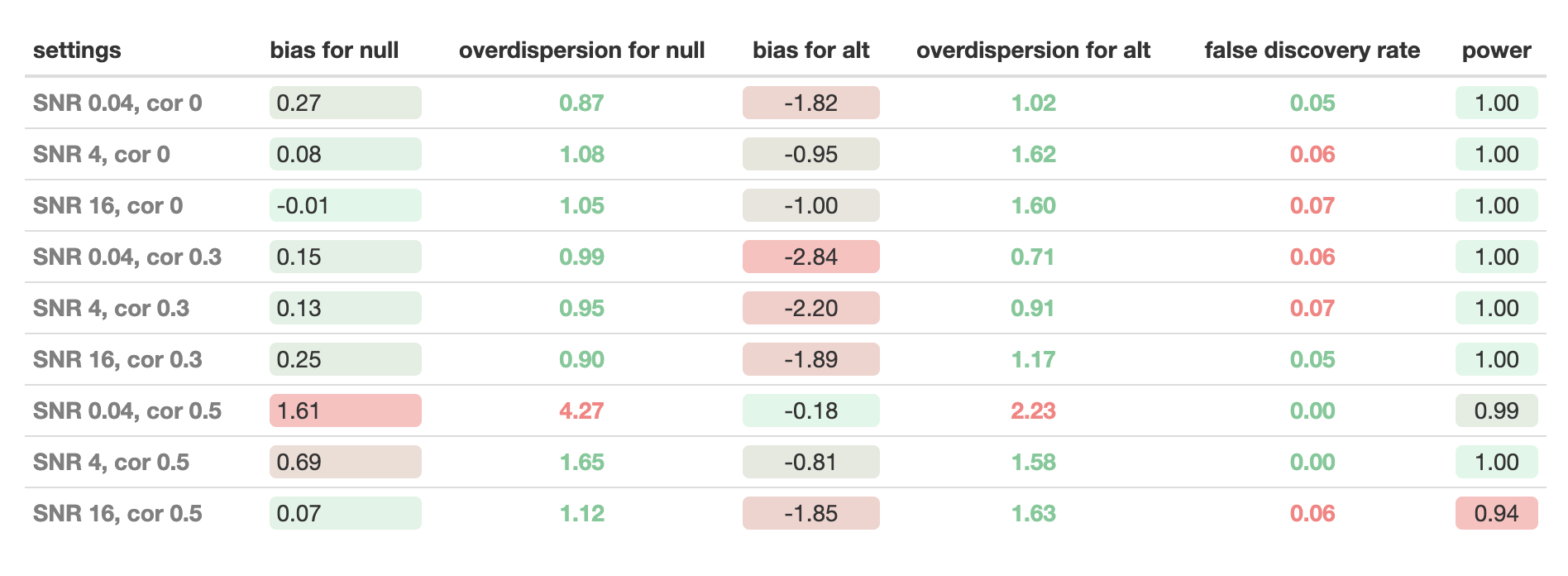}
    \caption{\em{Type-I Error and Power of Tests based on Parametric Bootstrap Procedure:}
    \small For \textit{bias}, the color shade gives a spectrum of bias values, with green indicating small bias, red indicating large bias; For \textit{overdispersion}, the text color gives a measurement of variance ratio, with red indicating large overdispersion ($>2$);  For \textit{FDR}, the text color indicates whether type-I error is controlled under nominal level ($\alpha=0.05$); For \textit{power}, the color shade represents the rejection rate when alternative is true.}
    \label{fig:power}
\end{table}

\section{Application to Genome-Wide Association Studies (GWAS)} \label{sec:GWAS}

\subsection{Prediction Model}
Heritability, the percentage of variability that is due to genetic causes, is of wide interest to scientists. \citeasnoun**{yang2010common} estimated the genetic variance of human height explained by genome-wide SNPs using linear mixed models (LMM). \citeasnoun**{weissbrod2019maximum} studied heritability estimation when samples are taken from a well-selected subpopulation. In this section, we apply pre-validation and bootstrap-based hypothesis testing procedure in GWAS with synthetic data. Our goal is to assess the relative contribution of genetics, as compared to other known factors.

We adopt the linear mixed model approach, and first introduce the following notation:
\begin{itemize}
    \item $y \in \mathbb{R}^n$ - continuous phenotype (e.g. height), with $n = 300$ in our study
    \item $X \in \mathbb{R}^{n \times e}$ - matrix of fixed effect covariates (e.g. age, sex, nutrition, or known SNPs with big effects, $e=5$), with coefficients $\beta \in \mathbb{R}^e$
    \item $Z \in \mathbb{R}^{n \times p}$ - matrix of random effect covariates (e.g. SNPs of three levels with small effects, $p=5000$), with random coefficients $b \in \mathbb{R}^p$, assume $b_j \sim N(0,\sigma^2_b)$ are i.i.d.
    \item $\epsilon \sim N(0,\sigma^2_\epsilon I_n) $ - vector of i.i.d. errors
\end{itemize}

In the linear mixed model, we assume that all the fixed effect covariates have large linear influences on the response, while all the gene expression variables have small linear influences:
\begin{equation}\label{eq:LMM}
\begin{aligned}
    y&= X\beta+ Zb+ \epsilon\\
    y|Z,X &\sim \mathcal{N}(X\beta,G\sigma^2_g + \sigma^2_\epsilon I_n)
\end{aligned}
\end{equation}
where $G = \frac{ZZ^T}{p} \in  \mathbb{R}^{n \times n}, \sigma^2_g = p\sigma^2_b \in  \mathbb{R}$.

We generate SNPs according to minor allele frequency (MAF), where MAF $\overset{\text{i.i.d.}}{\sim} Unif(0.15,0.85)$,  
\begin{equation*}
    Z_j = \begin{cases}
        0,&w.p.\quad(1-\text{MAF}_j)^2\\
        1,&w.p.\quad\text{MAF}_j(1-\text{MAF}_j)\\
        2,&w.p.{\quad\text{MAF}_j}^2
    \end{cases}
\end{equation*}

Then except for the first column being all 1, the fixed effect covariates and associated coefficients are generated according to:
\begin{equation*}
\begin{aligned}
    X = Z\Gamma+\eta,& \quad \eta\overset{\text{i.i.d.}}{\sim} N(0,1)\\
    \beta\overset{\text{i.i.d.}}{\sim}& N(0,0.3)
\end{aligned}
\end{equation*}

Lastly, $\Gamma$ is generated with a sparse structure according to a spike and slab model: 
\begin{equation*}
   \Gamma_{ij}\overset{\text{i.i.d.}}{\sim} \pi_\gamma \cdot N(0,0.1)+(1-\pi_\gamma)\cdot0,\quad \pi_\gamma \sim Bern(0.1)
\end{equation*}

Now, we are ready to apply the pre-validation procedure and  bootstrap-based hypothesis testing. \autoref{fig:GWASreuse} and \autoref{fig:GWASpv} summarize the result.
\begin{table}[H]
\centering
    \includegraphics[width=0.6\textwidth]{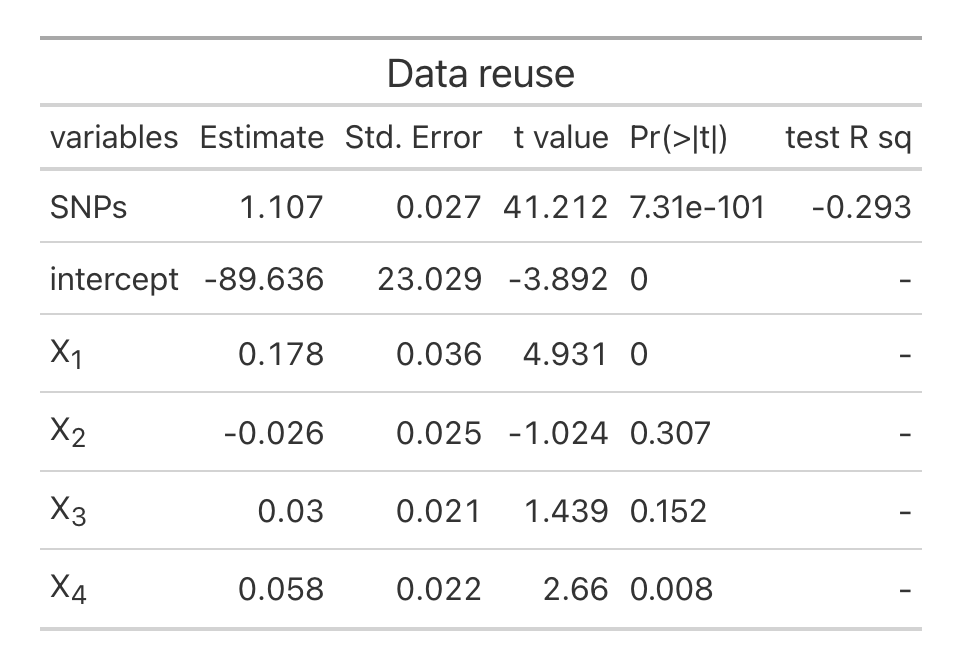}
    \caption{\em{Summary of GWAS Analysis and Inference from Data Reuse:}
    \small first-stage predictions fit with all observations, rather than leave-one-out, result in a significantly low p-value in second-stage inference; they also lead to bad model performance when generalizing to a set of test data, as the negative $R^2$ suggests serious overfitting and poor generalization.}
    \label{fig:GWASreuse}
\end{table}
\begin{table}[H]
\centering
   \includegraphics[width=0.9\textwidth]{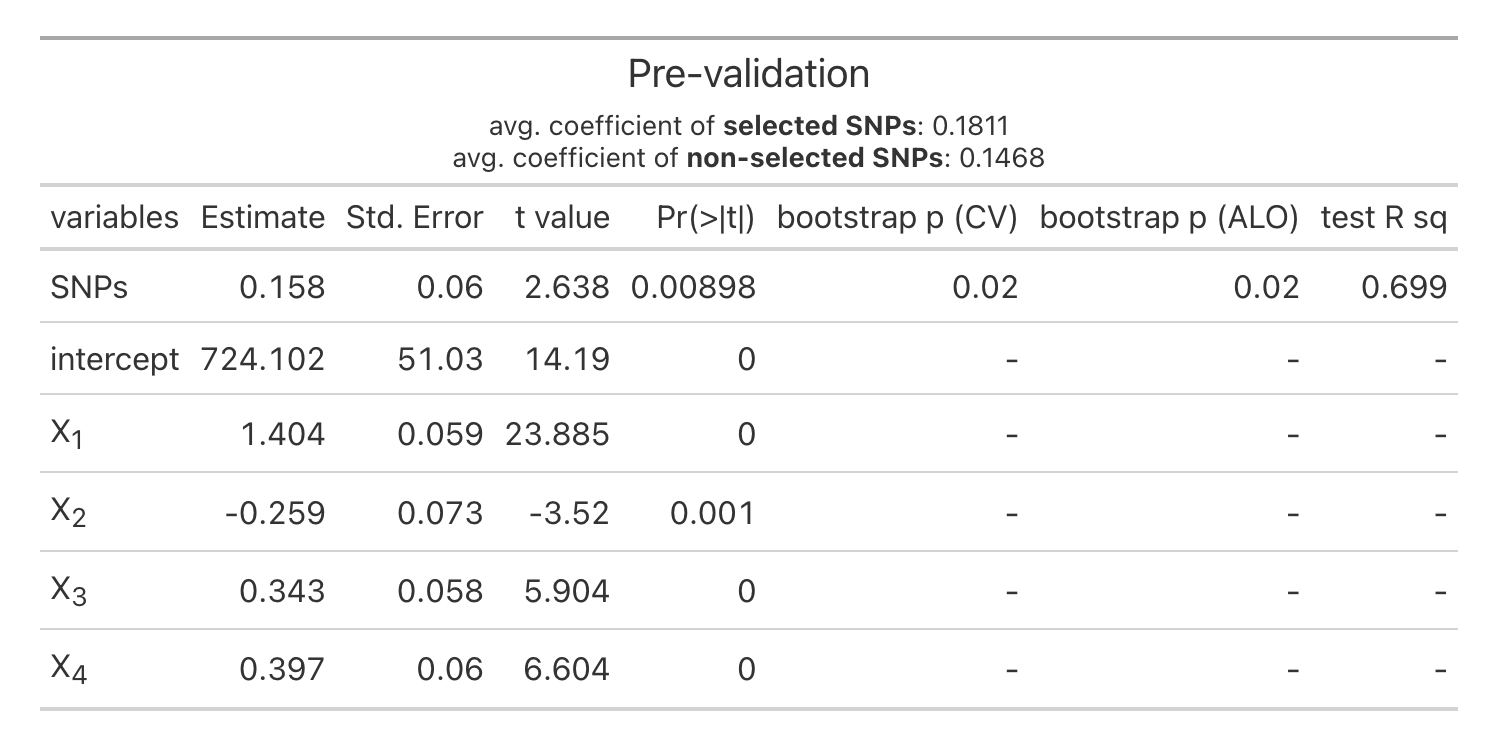}
   \caption{\em{Summary of GWAS Analysis and Inference from Pre-validation:}
   \small compared with the data-reuse procedure, pre-validation greatly reduces the significance of SNPs; our proposed bootstrap hypothesis test gives a more conservative estimate of p-value than the built-in R function, which uses standard Normal as the null; pre-validation also gives a better prediction model, suggested by the improvement of $R^2$ on a held-out test set.}
   \label{fig:GWASpv}
\end{table}
We can observe that pre-validation greatly reduces the significance of microarray predictor (named ``SNPs'' in the tables). Second, both \texttt{cv.glmnet()}-based and ALO-based bootstrap procedure give a more conservative estimate of p-value than \texttt{glm()} function, which uses Normal distribution as the null. Third, pre-validation significantly improved $R^2$ of the model on a held-out test set by 0.992, which strongly suggests it's a better model. Lastly, the average magnitude of coefficient $b$ from data generation for the selected set of SNPs is higher than for the non-selected set. This suggests that the Lasso regression is functioning properly and picking covariates with stronger signal.

\subsection{Estimation of Heritability }
Besides building a better prediction model, pre-validation also allows us to directly estimate heritability, defined as $\sigma^2_g = p\sigma^2_b \in  \mathbb{R}$. Here, we normalize $Var(y) = 1$, hence $\sigma^2_g$ gives a direct measurement of the proportion of variance explained by the SNPs.

As a benchmark for comparison, we first give a quick review on the classical restricted maximum likelihood (REML) method, proposed by \citeasnoun{patterson1971recovery}. Recall the formulation of linear mixed model \eqref{eq:LMM}, where 
\begin{equation}\label{eq:LMM_reg}
    Var(y|Z,X) = G\sigma^2_g + \sigma^2_\epsilon I_n
\end{equation} with $G = \frac{ZZ^T}{p}$ being the genetic relationship matrix between pairs of individuals at causal loci. If we know a priori the genotypes at the causal variants, we could fit model \eqref{eq:LMM_reg} and estimate the genetic variance $\sigma^2_g$. In the simulated data, we can assume all SNPs are causal variants. Hence, the REML procedure can be described as follows:

\begin{enumerate}[1)]
    \item Regress $y$ on $X$, get residuals $\Tilde{y} = y-\hat{y}$;
    \item Use observed genetic matrix $Z$ to compute $G = \frac{ZZ^T}{p}$;
    \item Compute $\Delta \Tilde{y}^2_{jk} = (\Tilde{y}_j - \Tilde{y}_k)^2$ for each pair of subjects, and regress on $G_{jk}$. Use the regression slope as an REML estimator of $-2\sigma_g^2$.
\end{enumerate}

On the other hand, pre-validation or data-reuse based estimator is constructed as follows:

\begin{enumerate}[1)]
    \item Regress $y$ on $X$, get residuals $\Tilde{y} = y-\hat{y}$;
    \item Regress $\Tilde{y}$ on $Z$ with relaxed Lasso regulation \cite{meinshausen2007relaxed}, get LOO fit $\hat{y}_{pv}$ or regular fit $\hat{y}_{reuse}$;
    \item Let $\frac{Var(\hat{y}_{pv})}{Var(\Tilde{y})}$ be a pre-validation estimator, or $\frac{Var(\hat{y}_{reuse})}{Var(\Tilde{y})}$ be a data-reuse estimator of $\sigma_g^2$.
\end{enumerate}

We recommend applying relaxed Lasso here, because coefficient shrinkage will result in underestimation of heritability, especially when the individual SNPs only have a small linear impact. We compare the classical REML approach with both pre-validation and data-reuse based estimators. \autoref{fig:table_heritability} and \autoref{fig:boxplot_heritability} show that the pre-validation based estimator is mostly on par with the REML estimator: 3\% higher bias, but improves standard deviation by 20\%, while the data-reuse based estimator has lowest standard deviation but highest bias.

\begin{table}[H]
    \centering
    \includegraphics[width=0.45\textwidth]{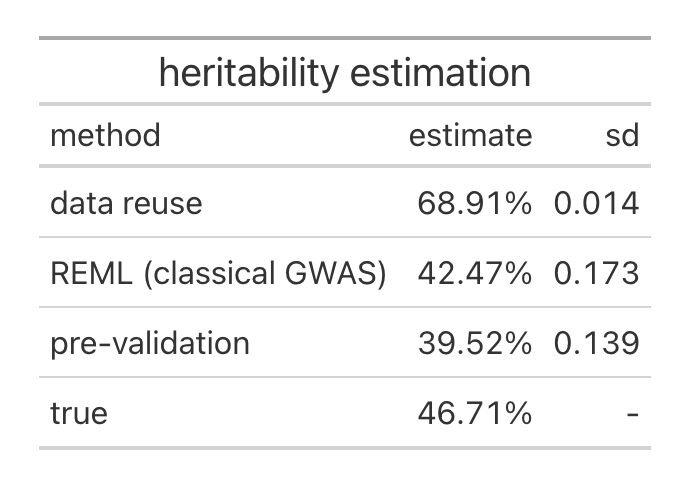}
    \caption{\em{Comparing Bias and Standard Deviation of Different Estimators for Heritability:}
    \small the pre-validation based estimator has 3\% higher bias than the REML estimator, but improves standard deviation by 20\%, while the data-reuse based estimator has the lowest standard deviation but a 30\% higher bias than the REML estimator.}
    \label{fig:table_heritability}
\end{table}

\begin{figure}[H]
    \centering
    \includegraphics[width=0.5\textwidth]{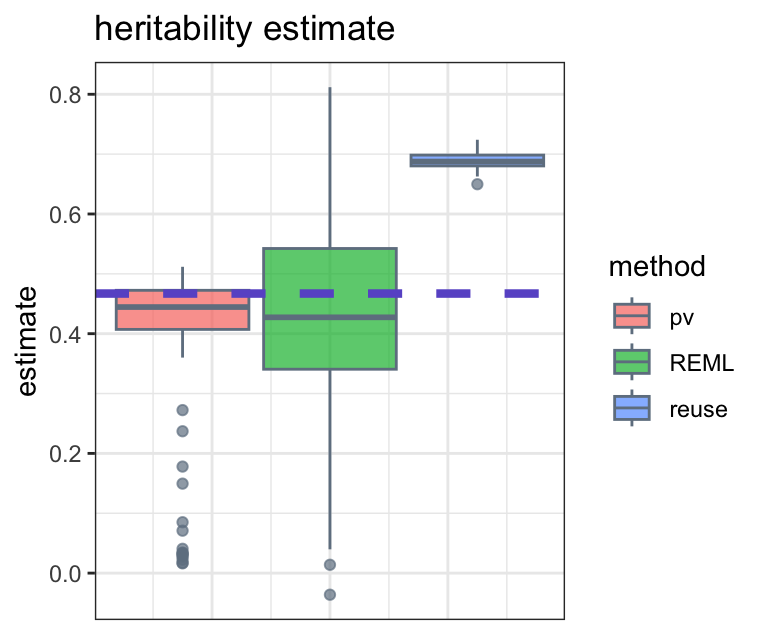}
    \caption{\em{Comparison Plot of Different Estimators for Heritability:}
    \small the pre-validation based estimator is mostly on par with the REML estimator in terms of bias, and outperforms in terms of variation; the data-reuse based estimator has unsatisfactory performance in both bias and variation.}
    \label{fig:boxplot_heritability}
\end{figure}

\section{Estimation of Mean Squared Error by Pre-validation}
\label{sec:error}
In this section, we investigate different estimates of  mean-squared error.  Along with the prediction and inference properties as discussed above, here we provide empirical evidence of how the pre-validation procedure also gives better error estimates than data-reuse. We compare both in-sample and out-of-sample test error, with the error estimates obtained from training set.

For the underlying data generating model, we consider two cases that differ with respect to whether the external predictors $X$ contribute to response $y$. The two models are as follows:

\begin{align}
    \text{Without externals: }
        y&= Z\beta_0 +\epsilon\\
    \text{With externals: }
        y&= X\beta_0 + Z\beta_1+ \epsilon\\
        Z&= X\Gamma + E \nonumber
\end{align}
where entries of $E \overset{\text{i.i.d.}}{\sim} N(0,\sigma^2_z)$, and entries
of $\epsilon \overset{\text{i.i.d.}}{\sim} N(0,\sigma^2_x)$, with $E$ and $\epsilon$ mutually independent. For in-sample test error, we fix $\mathbb{E}[y]$ and sample a new set of independent errors $\epsilon'$, then set $\mathbb{E}[y] + \epsilon'$ to be the new in-sample test responses. For out-of-sample test error, we draw new sets of predictors and independent errors, then generate new out-of-sample test responses according to the same rule.

Lastly, we define the following statistics of interest: 
\begin{itemize}
    \item $\Tilde{y}_{LOO}^{(1)}$ and $\Tilde{y}_{nonLOO}^{(1)}$: the predictor obtained from first-stage regression of $y$ on $Z$, based on either pre-validation or data-reuse procedure;
    \item  $\hat{y}_{LOO}^{(2)}$ and $\hat{y}_{nonLOO}^{(2)}$ (optional): the predictor obtained from second-stage regression of $y$ on $X$ and $\Tilde{y}^{(1)}$. This step is not necessary if no external predictors are included in the covariate set, because we do not need to regress y on the predictions of y again;
    \item $\frac{1}{n} \sum_{i=1}^n \left(y_i-\hat{y}_{i}^{(2)}\right)^2$: the error estimate computed based only on training data. If there is no external predictor, replace $\hat{y}_{i}^{(2)}$ with $\Tilde{y}_{i}^{(1)}$;
    \item $\frac{1}{n} \sum_{i=1}^n \left(y_{in,i}-\hat{y}_{i}^{(2)}\right)^2$: in-sample test error computed with the same prediction $\hat{y}_{i}^{(2)}$ as above. If there is no external predictor, replace $\hat{y}_{i}^{(2)}$ with $\Tilde{y}_{i}^{(1)}$;
    \item $\frac{1}{n} \sum_{i=1}^n \left(y_{out,i}-\hat{y}_{out,i}^{(2)}\right)^2$: out-of-sample test error computed with different prediction $\hat{y}_{out,i}^{(2)}$, as now covariates are drawn differently. We do not need to compute out-of-sample test error again for the pre-validation procedure when there is no external predictor, as the difference of first-stage models is negligible when $n$ is large.
\end{itemize}
 
Now \autoref{fig:err_noX} and \autoref{fig:err_withX} show the comparison of error estimates under different generating models. When there is no external predictor, the data-reuse error estimate tends to underestimate in-sample test error by $38.2\%$, and underestimate out-of-sample test error by $79\%$. On the other hand, the pre-validation error estimate is empirically unbiased for in-sample test error, and underestimates out-of-sample test error by only $1.3\%$.

When there are external predictors, the data-reuse error estimate underestimates in-sample test error by $60.2\%$, and even worse on out-of-sample test error. On the other hand, the pre-validation error estimate underestimates in-sample test error by $2\%$, and out-of-sample test error by $14.5\%$. Overall, we can conclude that the pre-validation procedure also provides a much better error estimate of the prediction model, compared with data-reuse.

\begin{table}[H]
    \centering
    \includegraphics[width=1\textwidth]{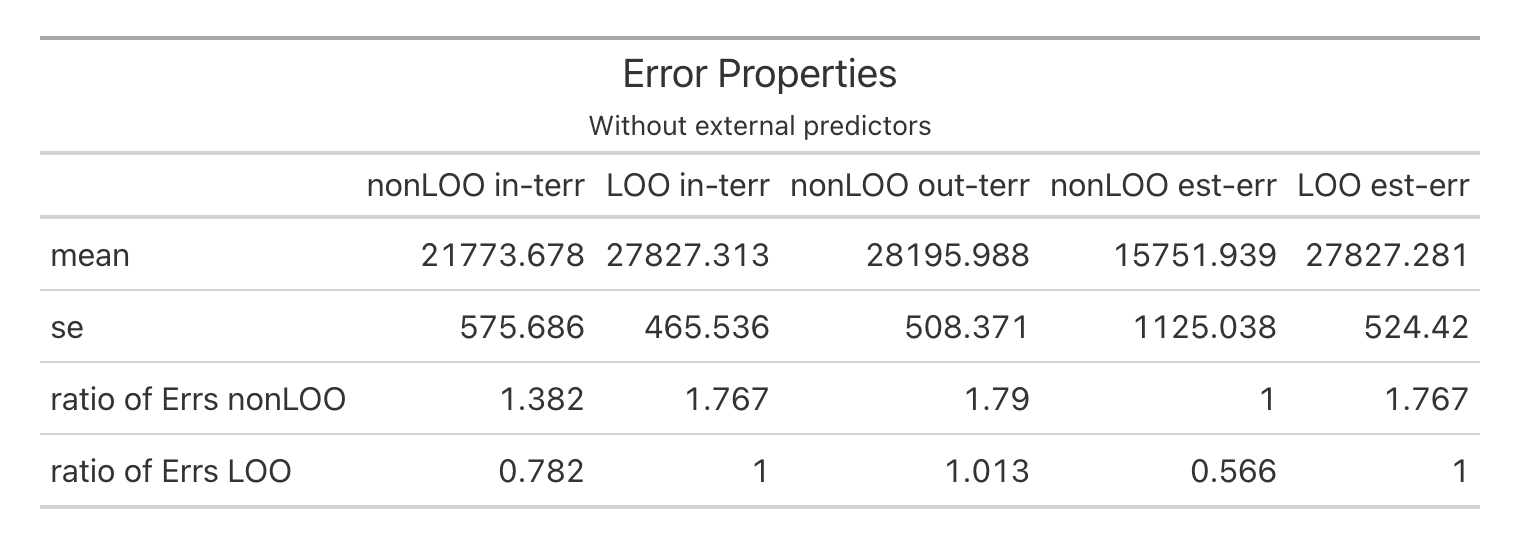}
    \caption{\em{Comparing LOO and non-LOO Error Estimates under no External Predictors}}
    \label{fig:err_noX}
\end{table}

\begin{table}[H]
    \centering
    \includegraphics[width=1\textwidth]{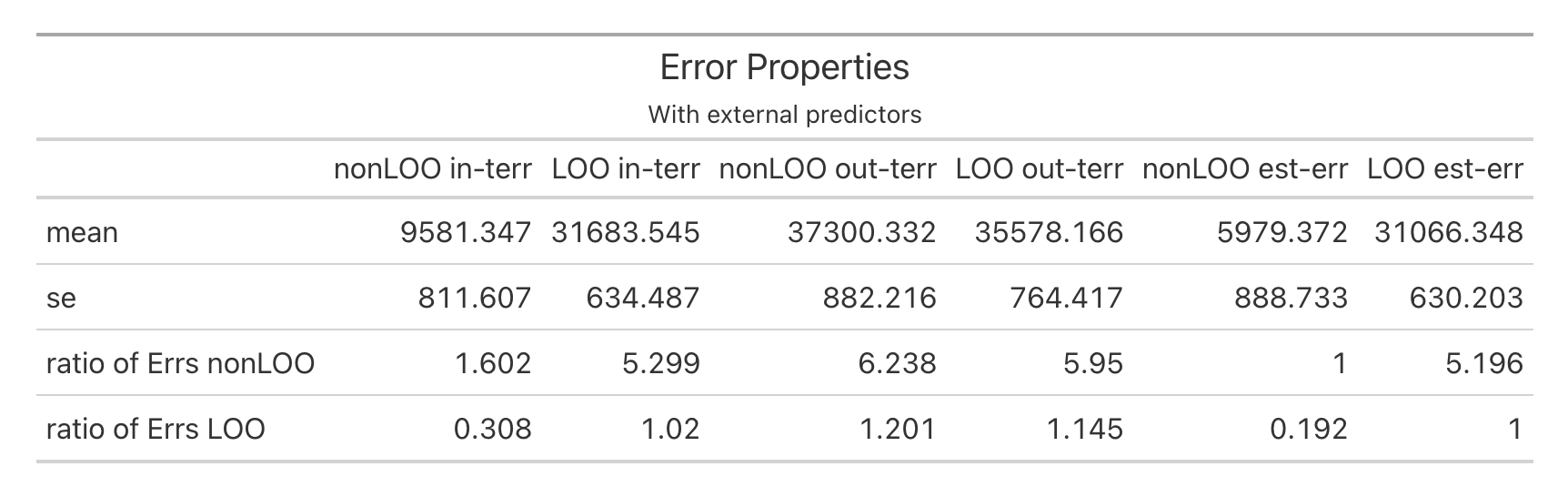}
    \caption{\em{Comparing LOO and non-LOO Error Estimates under External Predictors}}
    \label{fig:err_withX}
\end{table}

\section{Discussion}\label{sec:discussion}
In this paper, we revisited the pre-validation procedure for comparing external features with internal features that are derived from the data. We started with a more general formulation of the generating model, which does not impose independence assumption on the two covariate sets. We proposed not only an analytical distribution of the test statistics when least squares or ridge regression is applied to the first stage of pre-validation, but also a generic bootstrap distribution when the analysis is not tractable. We showed advantages of the pre-validation procedure in prediction, inference and error properties by simulation and various applications.

There are several directions  future research. First, the parametric bootstrap based hypothesis testing is not restricted to any parametric learning model. We are interested in developing an analogous approach when other non-parametric methods are applied to the first stage of pre-validation, such as the random forest. Second, we would hope that the proposed analytical distribution and bootstrap based procedure can transfer readily to the inference of external predictors. However, more careful investigations are expected. Last, even though there is strong empirical evidence in favor of the pre-validation procedure in terms of the error properties, further theoretical work remains to be done upon the topic.

\appendix
\section{Proof of \Cref{cor:equivalence}}\label{appendix:equivalence}
Assume $\Gamma = \mathbf{0}_{e \times p}$, $\Sigma = I_{e\times e}$, $\Theta = \mathbf{0}_{1 \times e}$, $\sigma^2_1 = ... \sigma^2_e = \sigma^2_x = 1$, $\beta_0 = \alpha_0 = \mathbf{0}_{e \times 1}$. Then,
\begin{gather*}
    B = \begin{bmatrix}
        1&&\mathbf{0}_{1\times p}\\
        \mathbf{0}_{p \times 1} && \sigma^2_zI_{p\times p}
    \end{bmatrix},\quad
    D= \mathbf{0}_{(p+1) \times e}\\
    (P_0,\, Q_0) \text{ has joint covariance matrix } \sigma^2_x \begin{bmatrix}
    B&&D\\
    D^T&&\Sigma \end{bmatrix} = \sigma^2_x \begin{bmatrix}
    1& & & \\[-0.2cm]
    & \sigma^2_z &\\[-0.2cm]
    & & \ddots &\\[-0.2cm]
    & & & \sigma^2_z &\\[-0.2cm]
    & & & &\Sigma \end{bmatrix}
\end{gather*}

Hence, \eqref{eq:linear} is reduced to:
\begin{equation*}
\begin{aligned}
    L&=\frac{(P_0 + D\alpha_0)^T B^{-1}(P_0-D\Sigma^{-1}Q_0)-\sigma_x^2 (p+1)}{\sigma_x\sqrt{(P_0+ D\alpha_0)^T (B^{-1}-B^{-1}D \Sigma^{-1}D^T B^{-1})(P_0 + D\alpha_0)}}\\
    &=\frac{P_0^T B^{-1}P_0-\sigma_x^2 (p+1)}{\sigma_x\sqrt{P_0^T B^{-1}P_0}}\\
    &=\frac{C-(p+1)}{\sqrt{C}}
\end{aligned}
\end{equation*}
where $P_0 \sim \mathcal{N}(0, \sigma^2_xB)$, $C \sim \chi_{p+1}^2$. The last equality follows from the fact that $B^{-1/2}P_0/\sigma_x \sim \mathcal{N}\left(0, I_{(p+1)\times(p+1)}\right)$, hence $P_0^T B^{-1}P_0/\sigma^2_x \sim \chi_{p+1}^2$.

This agrees with \eqref{t_noX}, where $t\xrightarrow{d} \frac{C-p}{\sqrt{C}}$, $C\sim \chi_p^2$, expect for we have $p+1$ in place of their $p$, since we have added an intercept to $Z$. \qed

\section{Proof of \Cref{thm:ridge}}\label{appendix:ridge}
\newcommand{\vol}{\mathrm{Vol}}
\newcommand{\I}{\mathrm{i}}
\newcommand{\mg}{\mathcal{G}}
\newcommand{\ep}{\epsilon}
\newcommand{\abs}[1]{| #1|}
\newcommand{\dt}[1]{{#1}_t}
\newcommand{\me}{\mathcal{E}}
\newcommand{\mk}{\mathcal{K}}
\newcommand{\mm}{\mathcal{M}}
\newcommand{\mr}{\mathcal{R}}
\newcommand{\MU}{\mathcal{U}}
\newcommand{\mc}{\mathcal{C}}
\newcommand{\ms}{\mathcal{S}}
\newcommand{\mw}{\mathcal{W}}
\newcommand{\tf}{\tilde{f}}
\newcommand{\tF}{\tilde{F}}
\newcommand{\tl}{\tilde{L}}
\newcommand{\hf}{\hat{f}}
\newcommand{\hq}{\widehat{Q}}

\newtheorem{thm}{Theorem}
\newtheorem{lmm}{Lemma}
\newtheorem{cor}{Corollary}
\newtheorem{prop}{Proposition}
\newtheorem{defn}{Definition}
\newtheorem{problem}{Problem}
\theoremstyle{definition}
\newtheorem{remark}{Remark}
\newtheorem{ex}{Example}

\newcommand{\avg}[1]{\bigl\langle #1 \bigr\rangle}
\newcommand{\bbb}{\mathbf{B}}
\newcommand{\bbg}{\mathbf{g}}
\newcommand{\bbr}{\mathbf{R}}
\newcommand{\bbw}{\mathbf{W}}
\newcommand{\bbx}{\mathbf{X}}
\newcommand{\bbxa}{\mathbf{X}^*}
\newcommand{\bbxb}{\mathbf{X}^{**}}
\newcommand{\bbxp}{\mathbf{X}^\prime}
\newcommand{\bbxpp}{\mathbf{X}^{\prime\prime}}
\newcommand{\bbxt}{\tilde{\mathbf{X}}}
\newcommand{\bby}{\mathbf{Y}}
\newcommand{\bbz}{\mathbf{Z}}
\newcommand{\bbzt}{\tilde{\mathbf{Z}}}
\newcommand{\bigavg}[1]{\biggl\langle #1 \biggr\rangle}
\newcommand{\bp}{b^\prime}
\newcommand{\bx}{\mathbf{x}}
\newcommand{\by}{\mathbf{y}}
\newcommand{\cc}{\mathbb{C}}
\newcommand{\cov}{\mathrm{Cov}}
\newcommand{\ddt}{\frac{d}{dt}}
\newcommand{\ee}{\mathbb{E}}
\newcommand{\fp}{f^\prime}
\newcommand{\fpp}{f^{\prime\prime}}
\newcommand{\fppp}{f^{\prime\prime\prime}}
\newcommand{\ii}{\mathbb{I}}
\newcommand{\gp}{g^\prime}
\newcommand{\gpp}{g^{\prime\prime}}
\newcommand{\gppp}{g^{\prime\prime\prime}}
\newcommand{\hess}{\operatorname{Hess}}
\newcommand{\ma}{\mathcal{A}}
\newcommand{\md}{\mathcal{D}}
\newcommand{\mf}{\mathcal{F}}
\newcommand{\mi}{\mathcal{I}}
\newcommand{\ml}{\mathcal{L}}
\newcommand{\cp}{\mathcal{P}}
\newcommand{\cq}{\mathcal{Q}}
\newcommand{\cs}{\mathcal{S}}
\newcommand{\mh}{\mathcal{H}}
\newcommand{\mx}{\mathcal{X}}
\newcommand{\mxp}{\mathcal{X}^\prime}
\newcommand{\mxpp}{\mathcal{X}^{\prime\prime}}
\newcommand{\my}{\mathcal{Y}}
\newcommand{\myp}{\mathcal{Y}^\prime}
\newcommand{\mypp}{\mathcal{Y}^{\prime\prime}}
\newcommand{\pp}{\mathbb{P}}
\newcommand{\mv}{\mathcal{V}}
\newcommand{\ppr}{p^\prime}
\newcommand{\pppr}{p^{\prime\prime}}
\newcommand{\rr}{\mathbb{R}}
\newcommand{\smallavg}[1]{\langle #1 \rangle}
\newcommand{\sss}{\sigma^\prime}
\newcommand{\st}{\sqrt{t}}
\newcommand{\sst}{\sqrt{1-t}}
\newcommand{\tr}{\operatorname{Tr}}
\newcommand{\uu}{\mathcal{U}}
\newcommand{\var}{\mathrm{Var}}
\newcommand{\ve}{\varepsilon}
\newcommand{\vp}{\varphi^\prime}
\newcommand{\vpp}{\varphi^{\prime\prime}}
\newcommand{\ww}{W^\prime}
\newcommand{\xp}{X^\prime}
\newcommand{\xpp}{X^{\prime\prime}}
\newcommand{\xt}{\tilde{X}}
\newcommand{\xx}{\mathcal{X}}
\newcommand{\yp}{Y^\prime}
\newcommand{\ypp}{Y^{\prime\prime}}
\newcommand{\zt}{\tilde{Z}}
\newcommand{\zz}{\mathbb{Z}}
\newcommand{\mn}{\mathcal{N}}
\newcommand{\upi}{\underline{\pi}}
\newcommand{\fb}{\mathfrak{B}}

\newcommand{\sign}{\operatorname{sign}}
 
\newcommand{\fpar}[2]{\frac{\partial #1}{\partial #2}}
\newcommand{\spar}[2]{\frac{\partial^2 #1}{\partial #2^2}}
\newcommand{\mpar}[3]{\frac{\partial^2 #1}{\partial #2 \partial #3}}
\newcommand{\tpar}[2]{\frac{\partial^3 #1}{\partial #2^3}}

\newcommand{\tT}{\tilde{T}}

\newcommand{\est}{e^{-t}}
\newcommand{\esst}{\sqrt{1-e^{-2t}}}
\newcommand{\bbu}{\mathbf{u}}
\newcommand{\bbv}{\mathbf{v}}
\newcommand{\bos}{{\boldsymbol \sigma}}
\newcommand{\bz}{\mathbf{z}}
\newcommand{\bg}{\mathbf{g}}
\newcommand{\av}[1]{\langle #1 \rangle}
\newcommand{\HP}{\hat{P}}
\newcommand{\hp}{\hat{p}}
\newcommand{\rank}{\operatorname{rank}}
\newcommand{\dist}{\operatorname{dist}}
\newcommand{\area}{\operatorname{area}}


\newcommand{\fs}{\mathbb{S}}
\newcommand{\fd}{\mathbb{D}}
\newcommand{\ftw}{\mathbb{T}}
\newcommand{\fst}{\mathbb{M}}
\newcommand{\mfn}{\mathfrak{N}}
\newcommand{\mfc}{\mathfrak{C}}
\newcommand{\mb}{\mathcal{B}}

\newcommand{\ts}{\tilde{s}}
\newcommand{\tp}{\tilde{\P}}

\newcommand{\hx}{\hat{x}}
\newcommand{\HX}{\hat{X}}
\newcommand{\tx}{\tilde{X}}
\newcommand{\ty}{\tilde{Y}}

\renewcommand{\Re}{\operatorname{Re}}

\newcommand{\dd}{\mathrm{d}}
\newcommand{\eq}[1]{\begin{align*} #1 \end{align*}}
\newcommand{\eeq}[1]{\begin{align} \begin{split} #1 \end{split} \end{align}}

\newcommand{\tv}{d_{\textup{TV}}}
\newcommand{\gammap}{\Gamma^{\textup{p}}}
\newcommand{\phip}{\Phi^{\textup{p}}}
\newcommand{\mfp}{\mathcal{F}^{\textup{p}}}
\newcommand{\muap}{\mu_a^{\textup{p}}}
\newcommand{\psip}{\Psi^{\textup{p}}}
\newcommand{\mgp}{\mathcal{G}^{\textup{p}}}

\newcommand{\ham}{\operatorname{H}}

\newtheorem{claim}[thm]{Claim}
\newtheorem{prob}[thm]{Problem}
\newtheorem{rmk}[thm]{Remark}
\newcommand{\N}{\mathbb{N}}
\newcommand{\Z}{\mathbb{Z}}
\newcommand{\Q}{\mathbb{Q}}
\newcommand{\R}{\mathbb{R}}
\newcommand{\e}{\mathrm{e}}
\newcommand{\1}{\mathbbm{1}}
\newcommand{\C}{\mathbb{C}}
\newcommand{\F}{\mathbb{F}}
\newcommand{\E}{\mathbb{E}}
\newcommand{\X}{\mathcal{X}}
\newcommand{\Xx}{\mathbb{X}}
\newcommand{\J}{\mathbb{J}}
\newcommand{\Eff}{\mathcal{F}}
\renewcommand{\bar}{\overline}
\renewcommand{\P}{\mathbb{P}}
\newcommand{\ul}{\underline}
\newcommand{\la}{\langle}
\newcommand{\ra}{\rangle}
\renewcommand{\d}{\mathrm{d}}
\renewcommand{\S}{\mathbb{S}}
\newcommand{\Ess}{\mathcal{S}}
\renewcommand{\tilde}{\widetilde}
\renewcommand{\hat}{\widehat}
\newcommand{\un}{\underline{n}}

\newcommand{\mz}{\mathcal{Z}}

\newcommand{\sij}{\smallavg{\sigma_i\sigma_j}}
\newcommand{\sjk}{\smallavg{\sigma_j\sigma_k}}
\newcommand{\skl}{\smallavg{\sigma_k\sigma_l}}
\newcommand{\sijkl}{\smallavg{\sigma_i\sigma_j\sigma_k \sigma_l}}
\newcommand{\cijkl}{\sijkl - \sij\skl}
\newcommand{\si}{\smallavg{\sigma_i}}
\newcommand{\sj}{\smallavg{\sigma_j}}
\newcommand{\sijk}{\smallavg{\sigma_i\sigma_j\sigma_k}}
\newcommand{\cijk}{\sijk - \si\sjk}
\newcommand{\sk}{\smallavg{\sigma_k}}

\newcommand{\cij}{\smallavg{\sigma_i\sigma_j} - \si\sj}
\newcommand{\cik}{\smallavg{\sigma_i\sigma_k} - \si\sk}
\newcommand{\tZ}{\tilde{Z}}
\newcommand{\tX}{\tilde{X}}



\setcounter{tocdepth}{2}






Throughout this section, we will work under the hypotheses of Theorem \cref{thm:ridge}. 
We begin the proof of Theorem \cref{thm:ridge} with the following lemma, which gives convenient expressions  for $\hat{\beta}_{PV}$ and $\hat{\sigma}(\hat{\beta}_{PV})$.
\begin{lmm}\label{mainlmm}
Define
\begin{align}\label{uvw}
u := \frac{1}{\sqrt{n}} \tilde{y}^T\epsilon, \ \ V := \frac{1}{n}X^T \tilde{y}, \ \ W :=\frac{1}{\sqrt{n}} X^T \epsilon,
\end{align}
and 
\begin{align}\label{rdef}
r := \frac{1}{\frac{1}{n}\tilde{y}^T\tilde{y} - V^T(\frac{1}{n}X^TX)^{-1}V}.
\end{align}
Then we have
\[
\sqrt{n}\hat{\beta}_{PV} = r (u - V^T ({\textstyle \frac{1}{n}}X^TX)^{-1} W).
\]
and $\hat{\sigma}(\hat{\beta}_{PV}) = \sigma \sqrt{r/n}$.
\end{lmm}
\begin{proof}
Note that
\begin{align}\label{eq1}
\tilde{X}^T y &= 
\begin{bmatrix}
\tilde{y}^T\\
X^T
\end{bmatrix} y
= 
\begin{bmatrix}
\tilde{y}^Ty\\
X^Ty
\end{bmatrix}= 
\begin{bmatrix}
\tilde{y}^T(X\beta_0+\epsilon)\\
X^T(X\beta_0 + \epsilon)
\end{bmatrix},
\end{align}
and 
\begin{align*}
\tilde{X}^T\tilde{X} &= 
\begin{bmatrix}
\tilde{y}^T\\
X^T
\end{bmatrix}
\begin{bmatrix}
\tilde{y} & X
\end{bmatrix}
= 
\begin{bmatrix}
\tilde{y}^T\tilde{y} & \tilde{y}^T X\\
X^T \tilde{y} & X^T X
\end{bmatrix}.
\end{align*}
Thus, defining $r$ as in the statement of the lemma and $A := (X^TX)^{-1}$, and using the well-known formula for the inverse of a block matrix, 
we get
\begin{align}\label{eq2}
(\tilde{X}^T\tilde{X})^{-1} &= \frac{r}{n}
\begin{bmatrix}
1 & - \tilde{y}^T XA\\
-AX^T \tilde{y} & A + AX^T\tilde{y}\tilde{y}^T XA
\end{bmatrix}.
\end{align}
From this identity, it is clear that $\hat{\sigma}(\hat{\beta}_{PV}) = \sigma \sqrt{r/n}$. Combining \eqref{eq1} and \eqref{eq2}, and letting $e_1 := (1,0,\ldots,0)^T$, we get
\begin{align*}
\hat{\beta}_{PV} &=  \frac{r}{n} e_1^T
\begin{bmatrix}
1 & - \tilde{y}^T XA\\
-AX^T \tilde{y} & A + AX^T\tilde{y}\tilde{y}^T XA
\end{bmatrix}
\begin{bmatrix}
\tilde{y}^T(X\beta_0+\epsilon)\\
X^T(X\beta_0 + \epsilon)
\end{bmatrix}\\
&= 
\frac{r}{n} \tilde{y}^T(I-XAX^T)(X\beta_0+\epsilon)\\
&= \frac{r}{n}\tilde{y}^T (I-XAX^T)\epsilon,
\end{align*}
where the last step holds because $(I - XAX^T)X = X - X(X^TX)^{-1}X^TX = 0$. It is easy to see that this is the same as the claimed identity.
\end{proof}

Thus, to understand the limiting behavior of $\sqrt{n}\hat{\beta}_{PV}$, we have to understand the limiting behaviors of $r$, $u$, $V$, and $W$ (and we already know that $(\frac{1}{n}X^TX)^{-1}\to\Sigma^{-1}$). We will now carry out these tasks. The first step is the following lemma. The matrices $B$ and $D$ defined below will be important later.

\begin{lmm}\label{zlmm}
As $n\to \infty$,
\begin{align*}
&\frac{1}{n}(Z^T Z+\lambda I) \stackrel{P}{\to} \Gamma^T \Sigma \Gamma +( \sigma_z^2 +\kappa) I, \ \ \ \frac{1}{n}Z^T X \stackrel{P}{\to} \Gamma^T\Sigma. 
\end{align*}
As a consequence,
\begin{align*}
&\frac{1}{n}(\tZ^T \tZ +\lambda I) \stackrel{P}{\to} \begin{bmatrix}
1+\kappa  & \Theta \Gamma\\
\Gamma^T \Theta^T & \Gamma^T \Sigma \Gamma + (\sigma_z^2+\kappa) I
\end{bmatrix} =: B, \\
&\frac{1}{n}\tZ^T X \stackrel{P}{\to}
\begin{bmatrix}
\Theta \\
\Gamma^T \Sigma 
\end{bmatrix} =: D.
\end{align*}
Lastly, the matrix $B$ defined above is nonsingular. 
\end{lmm}
\begin{proof}
From the definition of $Z$, we get
\begin{align}
\frac{1}{n}Z^TZ &= \frac{1}{n}(\Gamma^T X^T + E^T)(X\Gamma+E)\notag \\
&= \frac{1}{n}\Gamma^T X^T X\Gamma + \frac{1}{n}E^TX\Gamma + \frac{1}{n}\Gamma^TX^TE + \frac{1}{n}E^TE.\label{zexp}
\end{align}
Since $\frac{1}{n}X^TX \to \Sigma$, the first term on the right converges to $\Gamma^T \Sigma \Gamma$.
Since the elements of $E$ are i.i.d.~$N(0,\sigma_z^2)$, we get that for any $a,b\in  \R^p$, 
\[
a^TE^TX\Gamma b \sim N(0, \|a\|^2\|X\Gamma b\|^2). 
\]
Since 
\[
\frac{1}{n} \|X\Gamma b\|^2 = \frac{1}{n} b^T \Gamma^TX^T X \Gamma b \to b^T \Gamma^T \Sigma \Gamma b,
\]
this shows that $\frac{1}{n}a^TE^TX\Gamma b$ converges to $0$ in probability. Since this holds for any $a,b\in \R^p$, this shows that the second term on the right in \eqref{zexp} converges to the zero matrix in probability. The third term is just the transpose of the second, so that also goes to zero in probability. Finally, for the fourth term, notice that for any $1\le i, j\le p$, the $(i,j)^{\mathrm{th}}$ entry of $\frac{1}{n}E^TE$ is
\[
\frac{1}{n}\sum_{k=1}^n e_{ki} e_{kj},
\]
where $e_{kl}$ denotes the $(k,l)^{\mathrm{th}}$ entry of $E$. Clearly, this converges to zero (in probability) if $i\ne j$, and to $\sigma_z^2$ if $i=j$. This proves that $\frac{1}{n}E^TE \to \sigma_z^2 I$, thus completing the proof of the first identity claimed in the lemma.

Next, note that
\begin{align*}
\frac{1}{n}Z^T X &= \frac{1}{n}(\Gamma^T X^T + E^T)X= \frac{1}{n}\Gamma^TX^TX + \frac{1}{n}E^T X.
\end{align*}
The first term converges to $\Gamma^T \Sigma$. We claim that the second term converges in probability to the zero matrix of order $p\times e$. To see this, take any $a\in \R^p$ and $b\in \R^e$. Note that
\[
a^T E^T X b \sim N(0, \sigma_z^2 \|a\|^2 \|Xb\|^2). 
\]
Since 
\[
\frac{1}{n} \|Xb\|^2 = \frac{1}{n} b^T X^T X b \to b^T \Sigma b,
\]
this shows that $\frac{1}{n}a^TE^TX b$ converges to $0$ in probability. Since this holds for any $a,b$, this completes the proof of the second identity claimed in the lemma. 

To prove the third identity, note that
\begin{align*}
\frac{1}{n}\tZ^T \tZ &=\frac{1}{n}
\begin{bmatrix}
1^T\\
Z^T
\end{bmatrix}
\begin{bmatrix}
1 & Z
\end{bmatrix}= 
\begin{bmatrix}
1 & \frac{1}{n}1^TZ\\
\frac{1}{n}Z^T 1 & \frac{1}{n}Z^TZ
\end{bmatrix}.
\end{align*}
Note that
\begin{align*}
\frac{1}{n}1^TZ &= \frac{1}{n}1^T X \Gamma + \frac{1}{n}1^T E.
\end{align*}
By the law of large numbers, $\frac{1}{n}1^TE$ converges in probability to the zero matrix. By assumption, $\frac{1}{n}1^TX \to \Theta$, and we have already established the limit of $\frac{1}{n}Z^TZ$. This proves the third identity claimed in the lemma.

For the fourth identity, note that
\begin{align*}
\frac{1}{n}\tZ^T X &=
\frac{1}{n}
\begin{bmatrix}
1^T \\
Z^T 
\end{bmatrix} X
= \frac{1}{n}\begin{bmatrix}
1^TX \\
Z^T X
\end{bmatrix}.
\end{align*}
Since $\frac{1}{n}1^TX \to \Theta$ and $\frac{1}{n}Z^TX \to \Gamma^T\Sigma$ in probability, this completes the proof.

Lastly, to show that $B$ is nonsingular, note that by the well-known formula for the determinant of a block matrix, 
\begin{align*}
\det(B) &= 1 - \Theta \Gamma(\Gamma^T \Sigma \Gamma + \sigma_z^2 I)^{-1} \Gamma^T\Theta^T\\
&= \lim_{n\to \infty} \biggl(1 - \frac{1}{n^2}1^TX \Gamma(\Gamma^T \Sigma \Gamma + \sigma_z^2 I)^{-1} \Gamma^TX^T 1\biggr)\\
&= \lim_{n\to \infty} \frac{1}{n}1^T\biggl(I - \frac{1}{n}X \Gamma(\Gamma^T \Sigma \Gamma + \sigma_z^2 I)^{-1} \Gamma^TX^T\biggr)1.
\end{align*}
Now, note that the eigenvalues of the positive definite matrix $\frac{1}{n}X \Gamma(\Gamma^T \Sigma \Gamma + \sigma_z^2 I)^{-1} \Gamma^TX^T$ are the same as its singular values. The top $p$ singular values of this matrix are the same as the singular values of $ (\Gamma^T \Sigma \Gamma + \sigma_z^2 I)^{-1} \Gamma^T(\frac{1}{n}X^TX)\Gamma$, and the rest are zero. But this latter matrix converges to $ (\Gamma^T \Sigma \Gamma + \sigma_z^2 I)^{-1} \Gamma^T\Sigma\Gamma$ as $n\to\infty$. By spectral decomposition, it is clear that the singular values of this matrix are all strictly less than $1$. This shows that there is some $\delta >0$ such that for all large enough $n$, the eigenvalues of $I - \frac{1}{n}X \Gamma(\Gamma^T \Sigma \Gamma + \sigma_z^2 I)^{-1} \Gamma^TX^T$ are all bigger than $\delta$. This implies, by the above formula for $\det(B)$, that $\det(B)>0$.
\end{proof}

 Recall that $\tZ_{(i)}$ denotes the matrix $\tZ$ after omitting the $i^{\mathrm{th}}$ row. Define $X_{(i)}$, $\epsilon_{(i)}$  and $y_{(i)}$ similarly. Let $\|M\|$ denote the $\ell^2$ operator norm of a matrix $M$ --- that is, $\|M\|$ is the supremum of $\|Mx\|$ over all $x\in \R^d$ with $\|x\|=1$, where $d$ is the number of columns of $M$. Recall that $\|M\|$ is equal to the largest singular value of $M$, which implies that $\|M\| = \|M^T\|$.
\begin{lmm}\label{zlmm1}
As $n\to \infty$, we have $\|(\tZ^T\tZ+\lambda I)^{-1}\| = O_P(n^{-1})$, and
\[
\max_{1\le i\le n} \|(\tZ^T\tZ+\lambda I )^{-1}-(\tZ_{(i)}^T \tZ_{(i)}+\lambda I)^{-1}\| = o_P(n^{-1}). 
\]
\end{lmm}
\begin{proof}
Note that
\[
\tZ_{(i)}^T\tZ_{(i)} +\lambda I = \tZ^T\tZ+\lambda I - \tZ_i^T\tZ_i,
\]
and thus, by the well-known formula for the inverse of a rank-one perturbation, 
\begin{align}\label{zmain}
(\tZ_{(i)}^T\tZ_{(i)}+  \lambda I)^{-1} &= (\tZ^T\tZ+ \lambda I)^{-1} + \frac{(\tZ^T\tZ+ \lambda I)^{-1}\tZ_i^T\tZ_i(\tZ^T\tZ+ \lambda I)^{-1}}{1-\tZ_i(\tZ^T\tZ+ \lambda I)^{-1}\tZ_i^T}.
\end{align}
Let $B$ be as in Lemma \ref{zlmm}. Then by Lemma \ref{zlmm}, 
\begin{align}\label{zzeq}
\frac{1}{n}(\tZ^T\tZ + \lambda I)= B + o_P(1),
\end{align}
where $o_P(1)$ denotes a random $p\times p$ matrix whose entries are all $o_P(1)$ as $n\to\infty$. Since $B$ is nonsingular, equation~\eqref{zzeq} implies that 
\[
\biggl(\frac{1}{n}(\tZ^T\tZ +  \lambda I)\biggr)^{-1} = B^{-1} + o_P(1).
\]
In particular, this implies that $\|(\frac{1}{n}(\tZ^T\tZ+  \lambda I))^{-1}\| = O_P(1)$, proving the first claim of the lemma. Next, since the entries of $E$ are i.i.d.~$N(0,\sigma_z^2)$, standard results about maxima of Gaussian random variables imply that $\max_{1\le i\le n}\|E_i\|=O_P(\sqrt{\log n})$, where $E_i$ denotes the $i^{\mathrm{th}}$ row of $E$. Since $\tZ_i = [1 \ \ X_i \Gamma + E_i]$ and $\max_{1\le i\le n}\|X_i\|=o(\sqrt{n})$ by assumption, this implies that 
\begin{align}\label{zimax}
\max_{1\le i\le n} \|\tZ_i\|=o_P(\sqrt{n}).
\end{align} 
Thus,
\begin{align*}
\max_{1\le i\le n} |\tZ_i(\tZ^T\tZ+  \lambda I)^{-1}\tZ_i^T| &\le \max_{1\le i\le n}\|\tZ_i\|^2\|(\tZ^T\tZ+  \lambda I)^{-1}\| \\
&= o_P(n)O_P(n^{-1})= o_P(1).
\end{align*}
Similarly,
\begin{align*}
\max_{1\le i\le n}\|(\tZ^T\tZ+  \lambda I)^{-1}\tZ_i^T\tZ_i(\tZ^T\tZ +  \lambda I)^{-1}\| &\le \max_{1\le i\le n}\|\tZ_i (\tZ^T\tZ +  \lambda I)^{-1}\|^2\\
&\le \max_{1\le i\le n} \|\tZ_i\|^2 \|(\tZ^T\tZ+  \lambda I)^{-1}\|^2 \\
&= o_P(n) O_P(n^{-2}) = o_P(n^{-1}). 
\end{align*}
Plugging these estimates into \eqref{zmain}, we get the second claim.
\end{proof}
\begin{lmm}\label{zlmm2}
As $n\to \infty$, we have $\|\tZ^TX\| = O_P(n)$, and
\[
\max_{1\le i\le n} \|\tZ^TX - \tZ_{(i)}^T X_{(i)}\| = o_P(n). 
\]
\end{lmm}
\begin{proof}
Note that by \eqref{zzeq}, 
\[
\|\tZ^T\| = \|\tZ\| = \sup_{x\in \R^p, \, \|x\|=1} \sqrt{x^T \tZ^T\tZ x} = O_P(\sqrt{n}).
\]
On the other hand, since $\frac{1}{n}X^TX \to \Sigma$, 
\[
\|X\| = \sup_{x\in \R^e, \, \|x\|=1} \sqrt{x^T X^TX x} = O(\sqrt{n}).
\]
This proves that $\|\tZ^TX\|\le \|\tZ^T\|\|X\| = O_P(n)$.  Next note that 
\begin{align*}
\max_{1\le i\le n}\|\tZ^T X - \tZ_{(i)}^T X_{(i)}\| &= \max_{1\le i\le n}\|\tZ_i^T X_i\|\le  \max_{1\le i\le n}\|\tZ_i\|\| X_i\|. 
\end{align*}
Since $\max_{1\le i\le n}\|\tZ_i\| = o_P(\sqrt{n})$ by equation \eqref{zimax}, and $\max_{1\le i\le n}\|X_i\| = o(\sqrt{n})$ by assumption, this proves the second claim of the lemma.
\end{proof}

For $1\le i\le n$, define
\begin{align}\label{nudef}
H_i :=  \tZ_i (\tZ_{(i)}^T \tZ_{(i)} + \lambda I)^{-1} \tZ_{(i)}^T, \ \ \  \nu_i := H_i X_{(i)}\beta_0,
\end{align}
and let $\nu$ denote the vector of the $\nu_i$'s. Then note that 
\begin{align}
\tilde{y}_i &= \tZ_i (\tZ_{(i)}^T \tZ_{(i)}+ \lambda I)^{-1} \tZ_{(i)}^T y_{(i)}\notag \\
&= H_i (X_{(i)}\beta_0 + \epsilon_{(i)})\notag \\
&= \nu_i + H_i \epsilon_{(i)}. \label{ydecomp}
\end{align}
Define 
\[
G := \tZ (\tZ^T \tZ+ \lambda I)^{-1} \tZ^TX,
\]
and let $G_i$ denote the $i^{\mathrm{th}}$ row of $G$. Define 
\[
\nu' := G\beta_0. 
\]
The following lemma shows that under the hypotheses of Theorem \cref{thm:ridge}, $\nu'$ gives a close approximation to $\nu$.
\begin{lmm}\label{nunu}
We have $\|\nu'-\nu\|=o_P(1)$ as $n\to \infty$.
\end{lmm}
\begin{proof}
Define
\begin{align}\label{hiprime}
H_i' := \tZ_i (\tZ^T \tZ+ \lambda I)^{-1} \tZ_{(i)}^T.
\end{align}
Then note that
\begin{align*}
&|(H_i X_{(i)} - H_i' X_{(i)})\beta_0| = |\tZ_i ((\tZ_{(i)}^T \tZ_{(i)}+ \lambda I)^{-1} - (\tZ^T\tZ+ \lambda I)^{-1}) \tZ_{(i)}^TX_{(i)}\beta_0|\\
&\le \|\tZ_i\|\|(\tZ_{(i)}^T \tZ_{(i)}+ \lambda I)^{-1} - (\tZ^T\tZ+ \lambda I)^{-1}\|\|\tZ_{(i)}^TX_{(i)}\|\|\beta_0\|\\
&\le \|\tZ_i\| \|\beta_0\| \max_{1\le j\le n} (\|(\tZ_{(j)}^T \tZ_{(j)}+ \lambda I)^{-1} - (\tZ^T\tZ+ \lambda I)^{-1}\|\|\tZ_{(j)}^TX_{(j)}\|).
\end{align*}
Thus, by Lemma \ref{zlmm1}, Lemma \ref{zlmm2}, Lemma \ref{zlmm}, and the assumption that $\|\beta_0\|=O(n^{-1/2})$, we get
\begin{align}
\sum_{i=1}^n ((H_i X_{(i)} - H_i' X_{(i)})\beta_0)^2 &\le O(n^{-1}) o_P(n^{-2})O_P(n^2) \sum_{i=1}^n \|\tZ_i\|^2\notag \\
&= o_P(n^{-1}) \tr(ZZ^T) =  o_P(n^{-1}) \tr(Z^TZ) = o_P(1).\label{newnu1}
\end{align}
Next, note that
\begin{align*}
|(H_i' X_{(i)} - G_i)\beta_0| &= |\tZ_i (\tZ^T \tZ+ \lambda I)^{-1} (\tZ_{(i)}^TX_{(i)}- \tZ^TX) \beta_0|\\
&= |\tZ_i (\tZ^T \tZ+ \lambda I)^{-1} \tZ_i^TX_i \beta_0|\\
&= |\tZ_i (\tZ^T \tZ+ \lambda I)^{-1} \tZ_i^T| |X_i \beta_0|\\
&\le |\tZ_i (\tZ^T \tZ+ \lambda I)^{-1} \tZ_i^T|\|X_i\|\|\beta_0\|.
\end{align*}
Now recall that in Theorem \cref{thm:ridge}, the assumption about the maximum norm of $\|X_i\|$ is strengthened to $\max_{1\le i\le n} \|X_i\|=o(n^{1/3})$. Using this and following the same steps as in the derivation of \eqref{zimax}, we now get
\begin{align}\label{zimax2}
\max_{1\le i\le n} \|\tZ_i\|=o_P(n^{1/3}).
\end{align} 
Thus, by Lemma \ref{zlmm1} and the fact that $\|\beta_0\|=O(n^{-1/2})$, we get
\begin{align}
\sum_{i=1}^n ((H_i' X_{(i)} - G_i)\beta_0)^2 &= o(n^{2/3})O(n^{-1}) \sum_{i=1}^n (\tZ_i (\tZ^T \tZ+ \lambda I)^{-1} \tZ_i^T)^2\notag \\
&\le o(n^{-1/3}) n \|(\tZ^T\tZ+ \lambda I)^{-1}\|^2\max_{1\le i\le n}\|\tZ_i\|^4 \notag \\
&= o(n^{2/3}) O_P(n^{-2}) o_P(n^{4/3}) = o_P(1).\label{newnu2}
\end{align}
Combining \eqref{newnu1} and \eqref{newnu2}, we get the desired result.
\end{proof}

Next, define
\[
\delta_i := H_i \epsilon_{(i)},
\]
and let $\delta$ be the vector of the $\delta_i$'s. Define 
\[
\delta' := \tZ(\tZ^T\tZ+ \lambda I)^{-1}\tZ^T\epsilon.
\]
The following lemma shows that under the hypotheses of Theorem \cref{thm:ridge}, $\delta'$ gives a close approximation of $\delta$.
\begin{lmm}\label{deltadelta}
As $n\to \infty$, $\|\delta-\delta'\|=o_P(1)$.
\end{lmm}
\begin{proof}
Let $H_i'$ be defined as in equation \eqref{hiprime}. Then note that 
\begin{align*}
&\E[((H_i-H_i')\epsilon_{(i)})^2|E] = \sigma^2\|H_i - H_i'\|^2\\
&= \sigma^2\|\tZ_i ((\tZ_{(i)}^T \tZ_{(i)}+ \lambda I)^{-1} - (\tZ^T\tZ+ \lambda I)^{-1}) \tZ_{(i)}^T\|^2\\
&\le \sigma^2 \|\tZ_i\|^2 \max_{1\le j\le n} \|(\tZ_{(i)}^T \tZ_{(i)}+ \lambda I)^{-1} - (\tZ^T\tZ+ \lambda I)^{-1})\|^2 \|\tZ_{(i)}^T\|^2.
\end{align*}
Note that
\[
\|\tZ_{(i)}^T\|^2 = \|\tZ_{(i)}^T\tZ_{(i)}\| = \|\tZ^T \tZ - \tZ_i^T \tZ_i\|\le \|\tZ^T\tZ\|,
\]
where the inequality holds because 
\[
x^T(\tZ^T \tZ - \tZ_i^T \tZ_i)x \le x^T \tZ^T \tZ x
\]
for any $x\in \R^{p+1}$. 
Thus, by Lemma \ref{zlmm} and  Lemma \ref{zlmm1}, 
\begin{align}
\sum_{i=1}^n\E[((H_i-H_i')\epsilon_{(i)})^2|E] &\le o_P(n^{-2})O_P(n)\sum_{i=1}^n \|\tZ_i\|^2\notag\\
&= o_P(n^{-1})\tr(\tZ\tZ^T) \notag \\
&= o_P(n^{-1}) \tr(\tZ^T \tZ) = o_P(n^{-1})O_P(n)=o_P(1).\label{delta1}
\end{align}
For ease of notation, let
\[
Q := \sum_{i=1}^n((H_i-H_i')\epsilon_{(i)})^2.
\]
By Chebychev's inequality, for any $\delta>0$,
\begin{align*}
\P(Q >\delta |E) \le \frac{\E(Q|E) }{\delta}.
\end{align*}
This implies that for any $\eta, \delta >0$, 
\begin{align*}
\P(Q > \delta) &\le \P(Q > \delta, \, \E(Q|E) \le \eta ) + \P(\E(Q|E) > \eta)\\
&\le \E[\P(Q> \delta |E) 1_{\{\E(Q|E)\le \eta  )\}}] + \P(\E(Q|E) > \eta)\\
&\le \E\biggl[\frac{\E(Q|E) }{\delta}1_{\{\E(Q|E)\le \eta )\}}\biggr] + \P(\E(Q|E) > \eta)\\
&\le \frac{\eta}{\delta} +\P(\E(Q|E) > \eta).
\end{align*}
But by equation \eqref{delta1}, $\E(Q|E)\to 0$ in probability. Thus, the above inequality shows that for any $\eta, \delta >0$,
\[
\limsup_{n\to \infty} \P(Q > \delta) \le \frac{\eta}{\delta}.
\]
Since $\eta$ is arbitrary, this shows that the left side is, in fact, equal to zero for any $\delta>0$. Thus, we get
\begin{align}\label{ytilde1}
\sum_{i=1}^n((H_i-H_i')\epsilon_{(i)})^2 = o_P(1).
\end{align}
Next, note that 
\begin{align*}
H_i'\epsilon_{(i)} - \delta_i' &= \tZ_i (\tZ^T\tZ+\lambda I)^{-1} (\tZ_{(i)}^T\epsilon_{(i)} - \tZ^T \epsilon) = \tZ_i (\tZ^T\tZ+\lambda I)^{-1}\tZ_i^T \epsilon_i.
\end{align*}
Thus,
\begin{align}
\sum_{i=1}^n\E((H_i'\epsilon_{(i)} - \delta_i' )^2 |E) &\le \sigma^2 \sum_{i=1}^n(\tZ_i (\tZ^T\tZ+\lambda I)^{-1}\tZ_i^T)^2\notag \\
&\le  \sigma^2 \|(\tZ^T \tZ+\lambda I)^{-1}\|^2 \sum_{i=1}^n \|\tZ_i\|^2\notag \\
&= O_P(n^{-2}) \tr(\tZ\tZ^T) \notag \\
&= O_P(n^{-2})\tr(\tZ^T\tZ) \notag \\
&= O_P(n^{-2}) O_P(n) = O_P(n^{-1}).\label{delta2}
\end{align}
Combining \eqref{delta1} and \eqref{delta2} gives the desired result.
\end{proof}
We also need the following two lemmas, which show that $\nu-\nu'$ is nearly orthogonal to $\epsilon$ and $\delta-\delta'$ is at an asymptotically constant angle to $\epsilon$.
\begin{lmm}\label{extra1}
As $n\to \infty$, $(\nu-\nu')^T\epsilon = o_P(1)$.
\end{lmm}
\begin{proof}
Note that $\nu-\nu'$ has no dependence on $\epsilon$. Thus 
\begin{align*}
\E[((\nu-\nu')^T\epsilon)^2|E] = \|\nu-\nu'\|^2.
\end{align*}
By Lemma \ref{nunu}, the right side is $o_P(1)$. By the technique used to derive \eqref{ytilde1}, this proves the lemma.
\end{proof}
\begin{lmm}\label{extra2}
As $n\to \infty$, $(\delta-\delta')^T\epsilon =-\sigma^2(p+1)+ o_P(1)$.
\end{lmm}
\begin{proof}
For ease of notation, let
\[
L_i := (\tZ_{(i)}^T \tZ_{(i)}+\lambda I)^{-1} - (\tZ^T\tZ+\lambda I)^{-1}.
\]
Then note that
\begin{align*}
\delta_i - \delta_i' &= \tZ_i L_i \tZ_{(i)}^T\epsilon_{(i)} + \tZ_i (\tZ^T\tZ+\lambda I)^{-1} (\tZ_{(i)}^T\epsilon_{(i)} - \tZ^T \epsilon)\\
&= \tZ_i L_i (\tZ^T\epsilon - \tZ_i^T\epsilon_i) - \tZ_i (\tZ^T\tZ+\lambda I)^{-1} \tZ_i^T\epsilon_i.
\end{align*}
Let's call the first term on the right $\gamma_i$ and the second term $\eta_i$. Then note that 
Thus,
\begin{align*}
\gamma^T \epsilon&= \epsilon^T\tZ \sum_{i=1}^n L_i\tZ_i^T\epsilon_i - \sum_{i=1}^n \tZ_i L_i \tZ_i^T\epsilon_i^2\\
&= \sum_{i, j=1}^n \epsilon_j\tZ_j L_i\tZ_i^T\epsilon_i- \sum_{i=1}^n \tZ_i L_i \tZ_i^T\epsilon_i^2\\
&= \sum_{1\le i\ne j\le n} \epsilon_j\tZ_j L_i \tZ_i^T\epsilon_i.
\end{align*}
Now, if $i\ne j$ and $i'\ne j'$, then 
\[
\E(\epsilon_i\epsilon_j \epsilon_{i'}\epsilon_{j'}) =
\begin{cases}
1 &\text{ if } \{i,j\}=\{i',j'\},\\
0 &\text{ otherwise.}
\end{cases}
\]
Thus,
\begin{align*}
\E[(\gamma^T\epsilon)^2|E] &\le \sum_{1\le i\ne j\le n} (\tZ_j L_i \tZ_i^T)^2  + \sum_{1\le i\ne j\le n} (\tZ_j L_i \tZ_i^T)(\tZ_i L_j \tZ_j^T)\\
&\le \sum_{i, j=1}^n \|L_i\|^2 \|\tZ_i\|^2\|\tZ_j\|^2  + \sum_{i,j=1}^n \|L_i\|\|L_j\|\|\tZ_i\|^2\|\tZ_j\|^2\\
&\le 2\biggl(\sum_{i=1}^n \|\tZ_i\|^2\biggr) \max_{1\le i\le n}\|L_i\|^2.
\end{align*}
By Lemma \ref{zlmm1}, $ \max_{1\le i\le n}\|L_i\|^2 = o_P(n^{-2})$. By Lemma \ref{zlmm},
\[
\sum_{i=1}^n \|\tZ_i\|^2 = \tr(\tZ^T\tZ) = O_P(n). 
\]
Thus,  by the technique for proving \eqref{ytilde1}, we get $\gamma^T\epsilon= o_P(1)$. Next, note that
\begin{align*}
\eta^T \epsilon &= \sum_{i=1}^n \tZ_i (\tZ^T\tZ+\lambda I)^{-1} \tZ_i^T\epsilon_i^2.
\end{align*}
Thus, 
\begin{align*}
\E(\eta^T\epsilon|E) &= \sigma^2  \sum_{i=1}^n \tZ_i (\tZ^T\tZ+\lambda I)^{-1} \tZ_i^T\\
&= \sigma^2 \sum_{i=1}^n \tr((\tZ^T\tZ+\lambda I)^{-1} \tZ_i^T\tZ_i)\\
&= \sigma^2 \tr((\tZ^T\tZ+\lambda I)^{-1}\tZ^T\tZ)\\
&= \sigma^2 (p+1) -\sigma^2 \lambda \tr((\tZ^T\tZ+\lambda I)^{-1}). 
\end{align*}
Applying Lemma \ref{zlmm} to the above expression, we get
\[
\E(\eta^T\epsilon|E) \stackrel{P}{\to} \sigma^2 (p+1) - \sigma^2 \kappa \tr(B^{-1}).
\]
On the other hand, by Lemma \ref{zlmm} and equation \eqref{zimax2},
\begin{align*}
\var(\eta^T \epsilon|E) &= 3\sigma^4 \sum_{i=1}^n (\tZ_i (\tZ^T\tZ+\lambda I)^{-1} \tZ_i^T)^2\\
&\le 3\sigma^4 \|(\tZ^T\tZ+\lambda I)^{-1}\|^2\sum_{i=1}^n\|\tZ_i\|^4\\
&\le 3\sigma^4 \|(\tZ^T\tZ+\lambda I)^{-1}\|^2 (\max_{1\le i\le n}\|\tZ_i\|^2)\sum_{i=1}^n\|\tZ_i\|^2\\
&= o_P(n^{-2}) o_P(n^{2/3}) O_P(n) = o_P(n^{-1/3}).
\end{align*}
This completes the proof.
\end{proof}
Define
\[
y' := \nu' + \delta'.
\]
The following lemma shows that $y'$ is a close approximation of $\tilde{y}$ and that $\|y'\|$ and $\|\tilde{y}\|$ are both of order $1$.
\begin{lmm}\label{ynorms}
As $n\to \infty$, $\|\tilde{y}-y'\|=o_P(1)$, $\|\tilde{y}\|=O_P(1)$, and $\|y'\|=O_P(1)$.
\end{lmm}
\begin{proof}
Since $\tilde{y}=\nu+\delta$ and $y' = \nu'+\delta'$, it follows from Lemma \ref{nunu} and Lemma \ref{deltadelta} that $\|\tilde{y}-y'\|=o_P(1)$. Next, note that 
\begin{align*}
\|\nu'\|^2 &= \beta_0^T X^T \tZ (\tZ^T \tZ+\lambda I)^{-1}\tZ^TX\beta_0\\
&= (\sqrt{n}\beta_0)^T({\textstyle\frac{1}{n}}X^T \tZ)({\textstyle\frac{1}{n}}\tZ^T \tZ+\lambda I)^{-1} ({\textstyle\frac{1}{n}}\tZ^T X)(\sqrt{n}\beta_0)
\end{align*}
By Lemma \ref{zlmm} and the assumption that $\sqrt{n}\beta_0\to \alpha_0$, it follows that the above expression is $O_P(1)$. Thus, $\|\nu'\|=O_P(1)$. Next, note that
\begin{align*}
\E[\|\delta'\|^2 |E] &= \sigma^2 \tr(\tZ(\tZ^T\tZ+\lambda I)^{-1}\tZ^T) \\
&= \sigma^2 \tr((\tZ^T\tZ+\lambda I)^{-1}\tZ^T\tZ) \\
&= \sigma^2(p+1)-\sigma^2\kappa \tr(B^{-1}) + o_P(1).
\end{align*}
Thus, by the argument used to derive \eqref{ytilde1}, $\|\delta'\|=O_P(1)$. Combining, we get $\|y'\|=O_P(1)$. Since $\|\tilde{y}-y'\|=o_P(1)$, this also proves that $\|\tilde{y}\|=O_P(1)$.
\end{proof}
Define two random variables 
\begin{align*}
U_0 &:= \tilde{y}^T(I-X(X^TX)^{-1}X^T)\epsilon,\\
V_0 &:=\tilde{y}^T(I-X(X^TX)^{-1}X^T)\tilde{y}.
\end{align*}
Define $U_0'$ and $V_0'$ to be the random variables obtained by replacing $\tilde{y}$ with $y'$ in the definitions of $U_0$ and $V_0$. The following lemma shows that $U_0'$ and $V_0'$ are close approximations of $U_0$ and $V_0$.
\begin{lmm}\label{u0v0lmm}
As $n\to \infty$, $U_0-U_0' = -\sigma^2(p+1) +\sigma^2\kappa \tr(B^{-1})+ o_P(1)$ and $V_0-V_0' = o_P(1)$.
\end{lmm}
\begin{proof}
Since $Q_X := I-X(X^TX)^{-1}X^T$ is a projection operator, it reduces Euclidean norm. Consequently, by Lemma \ref{ynorms}
\begin{align*}
|V_0-V_0'| &\le \|\tilde{y}-y'\|\|Q_X\tilde{y}\| + \|y'\|\|Q_X(\tilde{y}-y')\|\\
&\le  \|\tilde{y}-y'\|\|\tilde{y}\| + \|y'\|\|\tilde{y}-y')\| = o_P(1).
\end{align*}
Next, let $P_X := I-Q_X = X(X^TX)^{-1}X^T$. Note that 
\begin{align*}
\E[\|P_X\epsilon\|^2|E] &= \tr(P_X^TP_X) = \tr((X^TX)^{-1}(X^TX)) = e.
\end{align*}
Thus, by the method of proving \eqref{ytilde2}, $\|P_X\epsilon\|=O_P(1)$. Combining this with Lemma \ref{ynorms}, Lemma \ref{extra1} and Lemma \ref{extra2}, we get
\begin{align*}
U_0 - U_0' &= (\tilde{y}-y')^T\epsilon - (\tilde{y}-y')^TP_X\epsilon\\
&= (\nu-\nu')^T\epsilon + (\delta-\delta')^T \epsilon + O_P(\|\tilde{y}-y'\| \|P_X\epsilon\|)\\
&= -\sigma^2(p+1)+\sigma^2\kappa \tr(B^{-1})+ o_P(1).
\end{align*} 
This completes the proof of the lemma.
\end{proof}

We are  now ready to complete the proof of \cref{thm:ridge}.
\begin{proof}[Proof of \cref{thm:ridge}]
Let 
\[
P :=\frac{1}{\sqrt{n}} \tZ^T \epsilon, \ \ \ Q := \frac{1}{\sqrt{n}} X^T\epsilon.
\]
Then Lemma \ref{zlmm} and a simple argument via characteristic functions show that as $n\to \infty$, $(P,Q)\Longrightarrow (P_0, Q_0)$, where $(P_0, Q_0)$ is the pair defined in the statement of \cref{thm:ridge}. Now note that
\begin{align*}
U_0' &= (\nu' + \delta')^T  (I-X(X^TX)^{-1}X^T)\epsilon\\
&= (\beta_0^T X^T\tZ (\tZ^T \tZ+\lambda I)^{-1} \tZ^T+ \epsilon^T \tZ(\tZ^T\tZ+\lambda I)^{-1}\tZ^T)(I-X(X^TX)^{-1}X^T)\epsilon\\
&= (\sqrt{n}\beta_0)^T ({\textstyle\frac{1}{n}}X^T\tZ)({\textstyle\frac{1}{n}}(\tZ^T\tZ+\lambda I))^{-1}P \\
&\qquad - (\sqrt{n}\beta_0)^T ({\textstyle\frac{1}{n}}X^T\tZ)({\textstyle\frac{1}{n}}(\tZ^T\tZ+\lambda I))^{-1}({\textstyle\frac{1}{n}}\tZ^T X) ({\textstyle\frac{1}{n}}X^TX)^{-1}Q\\
&\qquad + P^T ({\textstyle\frac{1}{n}}(\tZ^T\tZ+\lambda I))^{-1}P - P^T ({\textstyle\frac{1}{n}}(\tZ^T\tZ+\lambda I))^{-1}({\textstyle\frac{1}{n}}\tZ^TX)({\textstyle\frac{1}{n}}X^TX)^{-1}Q.
\end{align*}
Similarly,
\begin{align*}
V_0'&= (\nu' + \delta')^T  (I-X(X^TX)^{-1}X^T)(\nu'+\delta')\\
&= (\beta_0^T X^T\tZ (\tZ^T \tZ+\lambda I)^{-1} \tZ^T+ \epsilon^T \tZ(\tZ^T\tZ+\lambda I)^{-1}\tZ^T)(I-X(X^TX)^{-1}X^T)\\
&\qquad \qquad \cdot (\tZ (\tZ^T \tZ+\lambda I)^{-1} \tZ^TX\beta_0+  \tZ(\tZ^T\tZ+\lambda I)^{-1}\tZ^T\epsilon)\\
&= (\sqrt{n}\beta_0)^T ({\textstyle\frac{1}{n}}X^T\tZ)({\textstyle\frac{1}{n}}(\tZ^T\tZ+\lambda I))^{-1}({\textstyle\frac{1}{n}}\tZ^TX)(\sqrt{n}\beta_0)\\
&\qquad - (\sqrt{n}\beta_0)^T ({\textstyle\frac{1}{n}}X^T\tZ)({\textstyle\frac{1}{n}}(\tZ^T\tZ+\lambda I))^{-1}({\textstyle\frac{1}{n}}\tZ^TX)({\textstyle\frac{1}{n}}X^TX)^{-1}\\
&\hskip2in \cdot ({\textstyle\frac{1}{n}}X^T\tZ)({\textstyle\frac{1}{n}}(\tZ^T\tZ+\lambda I))^{-1}({\textstyle\frac{1}{n}}\tZ^TX)(\sqrt{n}\beta_0)\\
&\qquad + 2 P^T({\textstyle\frac{1}{n}}(\tZ^T\tZ+\lambda I))^{-1} ({\textstyle\frac{1}{n}}\tZ^TX)(\sqrt{n}\beta_0)\\
&\qquad \qquad  - 2 P^T({\textstyle\frac{1}{n}}(\tZ^T\tZ+\lambda I))^{-1} ({\textstyle\frac{1}{n}}\tZ^TX) ({\textstyle\frac{1}{n}}X^TX)^{-1}({\textstyle\frac{1}{n}}X^T\tZ)\\
&\hskip2in \cdot({\textstyle\frac{1}{n}}(\tZ^T\tZ+\lambda I))^{-1}({\textstyle\frac{1}{n}}\tZ^TX)(\sqrt{n}\beta_0)\\
&\qquad + P^T({\textstyle\frac{1}{n}}(\tZ^T\tZ+\lambda I))^{-1}P \\
&\qquad \qquad - P^T({\textstyle\frac{1}{n}}(\tZ^T\tZ+\lambda I))^{-1}({\textstyle\frac{1}{n}}\tZ^TX)({\textstyle\frac{1}{n}}X^TX)^{-1}({\textstyle\frac{1}{n}}X^T\tZ)({\textstyle\frac{1}{n}}(\tZ^T\tZ+\lambda I))^{-1}P.
\end{align*}
By Lemma \ref{zlmm} and the fact that $(P, Q) \Longrightarrow (P_0,Q_0)$, these expressions show that as $n\to \infty$, $(U_0', V_0') \Longrightarrow (S, T)$, where 
\begin{align*}
S &= \alpha_0^T D^TB^{-1}P_0 - \alpha_0^T D^TB^{-1}D \Sigma^{-1}Q_0+ P_0^T B^{-1}P_0 - P_0^T B^{-1}D\Sigma^{-1}Q_0\\
&= (P_0+D\alpha_0)^T B^{-1}(P_0 - D\Sigma^{-1}Q_0),
\end{align*}
and 
\begin{align}
T &= \alpha_0^T D^TB^{-1}D\alpha_0- \alpha_0^T D^TB^{-1}D\Sigma^{-1}D^TB^{-1}D\alpha_0\notag \\
&\qquad + 2 P_0^TB^{-1} D\alpha_0 - 2 P_0^TB^{-1} D \Sigma^{-1}D^TB^{-1}D\alpha_0\notag \\
&\qquad + P_0^TB^{-1}P_0 - P_0^TB^{-1}D\Sigma^{-1}D^TB^{-1}P_0\notag \\
&= (P_0 + D\alpha_0)^T (B^{-1} - B^{-1}D\Sigma^{-1}D^T B^{-1}) (P_0+D\alpha_0)\notag \\
&= (P_0 + D\alpha_0)^T B^{-1}(B - D\Sigma^{-1}D^T)B^{-1} (P_0+D\alpha_0).\label{texp}
\end{align}
Since $B$ is nonsingular and $P_0\sim N_{p+1}(0, B)$, all components of $B^{-1}(P_0 + D\alpha_0)$ are nonzero with probability one. Now, note that 
\begin{align*}
B - D\Sigma^{-1} D^T &= \begin{bmatrix}
1+\kappa  & \Theta \Gamma\\
\Gamma^T \Theta^T & \Gamma^T \Sigma \Gamma + (\sigma_z^2+\kappa) I
\end{bmatrix} -
\begin{bmatrix}
\Theta \\
\Gamma^T \Sigma 
\end{bmatrix}\Sigma^{-1}
\begin{bmatrix}
\Theta^T & \Sigma \Gamma
\end{bmatrix}
\\
&=
\begin{bmatrix}
1+\kappa -\Theta\Sigma^{-1}\Theta^T & 0\\
0 & (\sigma_z^2+\kappa)I
\end{bmatrix}.
\end{align*}
Now, recall that $X$ has a column of all $1$'s. Thus, the projection of the vector $1$ onto the column space of $X$ is $1$ itself. That is, $X(X^TX)^{-1}X^T1 = 1$. Consequently,
\[
1 - \frac{1}{n} 1^TX(X^TX)^{-1}X^T1 = 1 - \frac{1}{n}1^T1 = 0.
\]
But the limit of the left side above as $n\to\infty$ is $1-\Theta \Sigma^{-1} \Theta^T$. Thus, 
\begin{align*}
B - D\Sigma^{-1} D^T  = \begin{bmatrix}
\kappa & 0\\
0 & (\sigma_z^2+\kappa)I
\end{bmatrix}.
\end{align*}
Using this in equation \eqref{texp} shows that $T > 0$ with probability one. Thus,
\[
\frac{U_0'}{\sigma\sqrt{V_0'}} \Longrightarrow \frac{S}{\sigma\sqrt{T}}. 
\]
But by Lemma \ref{mainlmm},
\begin{align*}
\frac{\hat{\beta}_{PV}}{\hat{\sigma}(\hat{\beta}_{PV}) } &= \frac{ \tilde{y}^T(I-X(X^TX)^{-1}X^T)\epsilon}{\sigma \sqrt{\tilde{y}^T(I-X(X^TX)^{-1}X^T)\tilde{y}}} = \frac{U_0}{\sigma\sqrt{V_0}}.
\end{align*}
By Lemma \ref{u0v0lmm}, this completes the proof.
\end{proof}


For ease of notation, let $Q:= (\Gamma^T \Sigma \Gamma + \sigma_z^2 I)$. Then note that
\begin{align*}
B^{-1} &= \frac{1}{\det(B)}
\begin{bmatrix}
1 & -\Theta \Gamma Q\\
-Q\Gamma^T \Theta^T & \Theta \Gamma Q\Gamma^T \Theta^T
\end{bmatrix}.
\end{align*}
Thus,
\begin{align*}
B^{-1}D &= 
\begin{bmatrix}
1 & -\Theta \Gamma Q\\
-Q\Gamma^T \Theta^T & \Theta \Gamma Q\Gamma^T \Theta^T
\end{bmatrix}
\begin{bmatrix}
\Theta \\
\Gamma^T \Sigma 
\end{bmatrix} = 
\begin{bmatrix}
\Theta - \Theta \Gamma Q \Gamma^T \Sigma\\
-Q\Gamma^T \Theta^T \Theta + \Theta \Gamma Q\Gamma^T \Theta^T\Gamma^T \Sigma 
\end{bmatrix}.
\end{align*}
To see this, let $P\Delta P^T$ be the spectral decomposition of the positive definite matrix $\Gamma^T \Sigma \Gamma$, where $P$ is orthogonal and $\Delta$ is diagonal. Then 
\begin{align*}
&B^{-1} - B^{-1}D\Gamma B^{-1} \\
&= (\Gamma^T \Sigma \Gamma + \sigma_z^2 I)^{-1}  - (\Gamma^T \Sigma \Gamma + \sigma_z^2 I)^{-1} \Gamma^T \Sigma \Gamma (\Gamma^T \Sigma \Gamma + \sigma_z^2 I)^{-1} \\
&= P(\Delta + \sigma_z^2I)^{-1} P^T -  P(\Delta  + \sigma_z^2I)^{-1} P^TP \Delta P^T  P(\Delta + \sigma_z^2 I)^{-1} P^T\\
&= P((\Delta + \sigma_z^2I)^{-1} - (\Delta + \sigma_z^2I)^{-2}\Delta) P^T.
\end{align*}
In the above expression, the matrix sandwiched between $P$ and $P^T$ is a diagonal matrix whose diagonal entries are all strictly positive. This shows that $B^{-1}-B^{-1}D\Gamma B^{-1}$ is a positive definite matrix. In particular, since $D\beta \ne 0$ by assumption, this shows that
\begin{align*}
\beta_0^T D^TB^{-1}D\beta_0 - \beta_0^T D^T B^{-1}D\Gamma B^{-1} D\beta_0 &= \beta_0^T D^T(B^{-1} - B^{-1}D\Gamma B^{-1} ) D\beta_0 > 0.
\end{align*}

\begin{align*}
\hat{\beta}_{PV} &= \frac{\tilde{y}^T(I-X(X^TX)^{-1}X^T)\epsilon}{\tilde{y}^T(I-X(X^TX)^{-1}X^T)\tilde{y}}, 
\end{align*}
and 
\begin{align*}
\hat{\sigma}(\hat{\beta}_{PV}) &= \frac{\sigma }{\sqrt{\tilde{y}^T(I-X(X^TX)^{-1}X^T)\tilde{y}}}.
\end{align*}
Thus,

The next step is to show that $\frac{1}{n}\tilde{y}^T\tilde{y}$ and $V$ converge in probability to deterministic limits. For that, we need the following lemma, which gives a decomposition of $\tilde{y}$ into a `large'  component which is a function of $E$, and a `small' component that is independent of the large component. 
\begin{lmm}\label{ylmm}
Define
\begin{align*}
\mu &:= \tZ(\tZ^T\tZ)^{-1}\tZ^TX\beta_0.
\end{align*}
Then we have, as $n\to\infty$, $\|\tilde{y}-\mu\|=o_P(\sqrt{n})$. 
\end{lmm}
To prove Lemma \ref{ylmm}, we need two preliminary lemmas.

\begin{thm}\label{mainthm}
Assume that $\sigma^2, \sigma_z^2 > 0$. Suppose that as $n\to\infty$, $\beta_0$, $\sigma^2$, $\sigma_z^2$, $e$, $p$, and $\Gamma$ remain fixed, and $X$ varies with $n$ in such a way that 
\begin{enumerate}
\item $\frac{1}{n}X^TX$ converges to an $e\times e$ positive definite matrix $\Sigma$, 
\item $\frac{1}{n}1^T X$ converges to an $1\times e$ matrix $\Theta$ (where $1$ denotes the vector of all $1$'s), and 
\item $\max_{1\le i\le n} \|X_i\|=o(\sqrt{n})$, where $X_i$ denotes the $i^{\mathrm{th}}$ row of $X$ and $\|X_i\|$ is its Euclidean norm. 
\end{enumerate}
Further, assume that the vectors $\Gamma^T\Theta^T$ and $\Gamma^T\Sigma \beta_0$ are not collinear. 
Then, as $n\to\infty$,
\[
\frac{\hat{\beta}_{PV}}{\hat{\sigma}(\hat{\beta}_{PV})} \Longrightarrow N(0,1).
\]
\end{thm}
The theorem implies, under mild conditions, that if the signal-to-noise ratio (SNR) (which is, roughly speaking, the ratio of $\|\beta_0\|^2$ and $\sigma^2$) is fixed, then $\hat{\beta}_{PV}/\hat{\sigma}(\hat{\beta}_{PV})$ converges to $N(0,1)$ as $n\to\infty$, no matter how small the SNR is. However, if the SNR is small, it may require a very large $n$ to see a good fit to $N(0,1)$. For understanding the behavior under this setting, we need a theorem where simultaneously $n\to \infty$ and $\beta_0\to 0$. The next theorem gives this.

We are now ready to prove Lemma \ref{ylmm}.
\begin{proof}[Proof of Lemma \ref{ylmm}]

In particular, $\tilde{y}_i - \nu_i$ is independent of $\epsilon$. Thus, the above formula shows that
\begin{align*}
\E[(\tilde{y}_i - \nu_i)^2 | E] &= \E[H_i \epsilon_{(i)}|E]\\
&= \E[H_i \epsilon_{(i)}\epsilon_{(i)}^T H_i^T|E]\\
&= H_i \E[\epsilon_{(i)}\epsilon_{(i)}^T ] H_i^T\\
&= H_i (\sigma_z^2 I) H_i^T\\
&= \sigma_z^2 H_i H_i^T = \sigma_z^2 \tZ_i (\tZ_{(i)}^T\tZ_{(i)})^{-1}\tZ_i^T.
\end{align*}
By Lemma \ref{zlmm1}, 
\begin{align*}
\max_{1\le i\le n} \|(\tZ_{(i)}^T\tZ_{(i)})^{-1}\| &\le \|(\tZ^T \tZ)^{-1}\| + \max_{1\le i\le n} \|(\tZ^T\tZ)^{-1} - (\tZ_{(i)}^T\tZ_{(i)})^{-1}\|\\
&= O_P(n^{-1}) + o_P(n^{-1}) = O_P(n^{-1}).
\end{align*}
Thus, the previous identity and Lemma \ref{zlmm} show that
\begin{align*}
\E[\|\tilde{y}-\nu\|^2|E] &= \sum_{i=1}^n \E[(\tilde{y}_i - \nu_i)^2 | E] \\
&\le \sigma_z^2 (\max_{1\le i\le n}\|(\tZ_{(i)}^T\tZ_{(i)})^{-1}\|) \sum_{i=1}^n \|\tZ_i\|^2\\
&= O_P(n^{-1}) \tr(\tZ\tZ^T)\\
&= O_P(n^{-1}) \tr(\tZ^T\tZ) = O_P(n^{-1}) O_P(n) = O_P(1).
\end{align*}
This implies that for any given $\delta>0$, we can choose $K$ large enough so that for all $n$,
\[
\P(\E[\|\tilde{y}-\nu\|^2|E] > K) < \delta.
\]
But by Chebychev's inequality,
\begin{align*}
\P(\|\tilde{y}-\nu\|^2 > K^2|E) \le \frac{\E[\|\tilde{y}-\nu\|^2|E] }{K^2}.
\end{align*}
Combining, we get
\begin{align*}
\P(\|\tilde{y}-\nu\|^2 > K^2) &\le \P(\|\tilde{y}-\nu\|^2 > K^2, \, \E[\|\tilde{y}-\nu\|^2|E] \le K) + \P(\E[\|\tilde{y}-\nu\|^2|E] > K)\\
&\le \E[\P(\|\tilde{y}-\nu\|^2 > K^2|E) 1_{\{\E[\|\tilde{y}-\nu\|^2|E] \le K)\}}] + \delta\\
&\le \E\biggl[\frac{\E[\|\tilde{y}-\nu\|^2|E] }{K^2}1_{\{\E[\|\tilde{y}-\nu\|^2|E] \le K)\}}\biggr] + \delta\\
&\le K^{-1}+\delta.
\end{align*}
Without loss of generality, we may choose $K > 1/\delta$. Then the above inequality implies that for all $n$,
\[
\P(\|\tilde{y}-\nu\|^2 > K^2) \le 2\delta.
\]
Therefore, 
\begin{align}\label{ytilde2}
\|\tilde{y}-\nu\| = O_P(1)
\end{align}
as $n\to \infty$. Next, note that for any $i$,
\begin{align*}
|\mu_i -\nu_i| &= \|\tZ_i(\tZ^T\tZ)^{-1}\tZ^TX\beta_0 - \tZ_i (\tZ_{(i)}^T \tZ_{(i)})^{-1} \tZ_{(i)}^T X_{(i)}\beta_0\|\\
&\le \|\tZ_i\|\|\beta_0\|(\|(\tZ^T\tZ)^{-1} - (\tZ_{(i)}^T \tZ_{(i)})^{-1}\|\|\tZ^T X\| \\
&\hskip2in + \|(\tZ_{(i)}^T \tZ_{(i)})^{-1}\|\|\tZ^TX - \tZ_{(i)}^T X_{(i)}\|).
\end{align*}
Thus, the estimates from Lemma \ref{zlmm1}, Lemma \ref{zlmm2} and Lemma \ref{zlmm} imply that
\begin{align}
\|\mu - \nu\|^2 &\le (o_P(n^{-1}) O_P(n) + O_P(n^{-1}) o_P(n)) \sum_{i=1}^n \|\tZ_i\|^2\notag \\
&= o_P(1) \tr(\tZ\tZ^T) = o_P(1)\tr(\tZ^T\tZ) = o_P(1)O_P(n) = o_P(n).\label{ytilde3}
\end{align}
Combining \eqref{ytilde2} and \eqref{ytilde3}, we get the desired result.
\end{proof}

We now have the ingredients necessary for computing the limits of the matrix $V$ defined in equation \eqref{uvw} and the number $r$ defined in equation \eqref{rdef}. This is the content of the following lemma.

\begin{lmm}\label{vrlmm}
Let $B$ and $D$ be defined as in Lemma \ref{zlmm}. 
Then, as $n\to \infty$,
\[
V \stackrel{P}{\to} D^T B^{-1} D\beta_0,
\]
and
\[
r\stackrel{P}{\to} \frac{1}{\beta_0^T D^TB^{-1}D\beta_0 - \beta_0^T D^T B^{-1}D\Sigma^{-1} D^TB^{-1} D\beta_0},
\]
where the denominator is strictly positive.
\end{lmm}
\begin{proof}
First, note that by Lemma \ref{ylmm},
\begin{align*}
\|X^T\tilde{y} - X^T \mu\| &\le \|X^T\|\|\tilde{y}-\mu\| = \|X^T\| o_P(\sqrt{n}).
\end{align*}
Now, $\|X^T\| = \|X\| = $ the largest singular value of $X$. This is equal to the square root of the  largest eigenvalue of $X^TX$. Since $\frac{1}{n}X^TX\to \Sigma$ and the largest eigenvalue is a continuous map on the space of $e\times e$ matrices, this shows that $\|X^T\|=O(\sqrt{n})$. Thus, by the above identity, we get that
\[
\|X^T\tilde{y} - X^T \mu\|  = O(\sqrt{n}) o_P(\sqrt{n}) = o_P(n).
\]
But $V=\frac{1}{n}X^T \tilde{y}$ and by Lemma \ref{zlmm},
\begin{align}
\frac{1}{n}X^T\mu &= \frac{1}{n}X^T \tZ(\tZ^T\tZ)^{-1}\tZ^TX\beta_0\notag\\
&\stackrel{P}{\to} D^TB^{-1} D\beta_0.\label{xmulim}
\end{align}
This proves the first claim of the lemma. Next, note that by the triangle inequality and Lemma \ref{ylmm},
\begin{align}
|\|\tilde{y}\|^2 - \|\mu\|^2| &\le |\|\tilde{y}\|-\mu\||(\|\tilde{y}\| + \|\mu\|)\notag\\
&\le \|\tilde{y}-\mu\| (\|\tilde{y}-\mu\| + 2\|\mu\|)\notag\\
&\le o_P(\sqrt{n}) ( o_P(\sqrt{n}) + 2\|\mu\|)\notag \\
&= o_P(n) + o_P(\sqrt{n})\|\mu\|.\label{ymu}
\end{align}
By Lemma \ref{zlmm},
\begin{align}
\frac{1}{n}\|\mu\|^2 &= \frac{1}{n}\mu^T \mu =\frac{1}{n}\beta_0^TX^T\tZ(\tZ^T\tZ)^{-1}\tZ^T \tZ(\tZ^T\tZ)^{-1}\tZ^TX\beta_0\notag \\
&= \beta_0^T({\textstyle\frac{1}{n}}X^T\tZ)({\textstyle\frac{1}{n}}\tZ^T\tZ)^{-1}({\textstyle\frac{1}{n}}\tZ^TX)\beta_0\notag \\
&\stackrel{P}{\to} \beta_0^T D^TB^{-1}D\beta_0.\label{mulim}
\end{align}
In particular, $\|\mu\|= O_P(\sqrt{n})$. Plugging this into \eqref{ymu}, we get that 
\[
\frac{1}{n}\|\tilde{y}\|^2 = \frac{1}{n}\|\mu\|^2 + o_P(1)= \beta_0^T D^TB^{-1}D\beta_0 + o_P(1).
\]
Combining this with the limit of $V$ that we already established, we get
\begin{align*}
r &= \frac{1}{\frac{1}{n}\tilde{y}^T\tilde{y} - V^T(\frac{1}{n}X^TX)^{-1}V}\\
&\stackrel{P}{\to}\frac{1}{\beta_0^T D^TB^{-1}D\beta_0 - \beta_0^T D^T B^{-1}D\Sigma^{-1} D^T B^{-1} D\beta_0}. 
\end{align*}
Thus, it only remains to show that the denominator is strictly positive. To see this, first note that 
\begin{align*}
B - D\Sigma^{-1} D^T &= \begin{bmatrix}
1 & \Theta \Gamma\\
\Gamma^T \Theta^T & \Gamma^T \Sigma \Gamma + \sigma_z^2 I
\end{bmatrix} -
\begin{bmatrix}
\Theta \\
\Gamma^T \Sigma 
\end{bmatrix}\Sigma^{-1}
\begin{bmatrix}
\Theta^T & \Sigma \Gamma
\end{bmatrix}
\\
&=
\begin{bmatrix}
1 -\Theta\Sigma^{-1}\Theta^T & 0\\
0 & \sigma_z^2I
\end{bmatrix}.
\end{align*}
Now, recall that $X$ has a column of all $1$'s. Thus, the projection of the vector $1$ onto the column space of $X$ is $1$ itself. That is, $X(X^TX)^{-1}X^T1 = 1$. Consequently,
\[
1 - \frac{1}{n} 1^TX(X^TX)^{-1}X^T1 = 1 - \frac{1}{n}1^T1 = 0.
\]
But the limit of the left side above as $n\to\infty$ is $1-\Theta \Sigma^{-1} \Theta^T$. Thus, 
\begin{align}\label{lucky}
B - D\Sigma^{-1} D^T  = \begin{bmatrix}
0 & 0\\
0 & \sigma_z^2I
\end{bmatrix}.
\end{align}
Next, note that 
\begin{align*}
B^{-1} &= \frac{1}{\det(B)}
\begin{bmatrix}
1 & -\Theta \Gamma Q\\
-Q\Gamma^T \Theta^T & Q + Q\Gamma^T\Theta^T\Theta \Gamma Q
\end{bmatrix}
\end{align*}
where $Q := (\Gamma^T \Sigma \Gamma + \sigma_z^2 I)^{-1}$. 
This gives us
\begin{align*}
B^{-1}D\beta_0 &= \frac{1}{\det(B)}
\begin{bmatrix}
\Theta \beta_0 -\Theta \Gamma Q \Gamma^T \Sigma \beta_0\\
-Q\Gamma^T \Theta^T\Theta \beta_0 + (Q + Q\Gamma^T\Theta^T\Theta \Gamma Q)\Gamma^T \Sigma\beta_0
\end{bmatrix}
\end{align*}
Since $Q$ is nonsingular, the last assumption in Theorem \ref{mainthm} proves that the lower component of $B^{-1}D\beta_0$ is a nonzero vector. By the formula \eqref{lucky}, this shows that the denominator in the limit of $r$ is strictly positive.
This completes the proof of the lemma.
\end{proof}

Our next goal is to identify the limiting joint distribution of $(u,W)$. As a first step, we will approximate $u$ by the random variable
\[
v := \frac{1}{\sqrt{n}}\mu^T\epsilon.
\]
The following lemma shows $v$ approximates $u$ in the limit.
\begin{lmm}\label{uvlmm}
As $n\to\infty$, $u-v\to 0$ in probability.
\end{lmm}
\begin{proof}
Recall the vector $\nu$ defined in equation \eqref{nudef}. Define
\[
w := \frac{1}{\sqrt{n}} \nu^T \epsilon.
\]
Note that $\mu$ and $\nu$ have no dependence on $\epsilon$. Since the entries of $\epsilon$ are i.i.d.~$N(0,\sigma^2)$ random variables, the above identity and the estimate \eqref{ytilde3} show that
\begin{align*}
\E[(v-w)^2|E] &= \frac{\sigma^2}{n}\|\mu-\nu\|^2 = o_P(1).
\end{align*}
From this, using an argument identical to the proof of \eqref{ytilde1}, it follows that $v-w=o_P(1)$. Thus, it suffices to show that $u-w \to 0$ in probability. To prove this, first note that
\begin{align*}
(u-w)^2 &= \frac{1}{n}\sum_{i,j=1}^n (\tilde{y}_i - \nu_i)(\tilde{y}_j - \nu_j)\epsilon_i\epsilon_j\\
&= \frac{1}{n}\sum_{i,j=1}^n (H_i \epsilon_{(i)})(H_j\epsilon_{(j)})\epsilon_i \epsilon_j,
\end{align*}
where $H_i$ is defined as in equation \eqref{nudef}.
Take any $1\le i\ne j\le n$. Let $H_{ij}$ denote the $j^{\mathrm{th}}$ component of the vector $H_i$ (when the components are indexed by the set $\{1,\ldots,n\}\setminus\{i\}$). Then note that the random variables $Q_{ij} := H_i\epsilon_{(i)} - H_{ij} \epsilon_j$ and $Q_{ji} := H_j\epsilon_{(j)} - H_{ji} \epsilon_i$ have no dependence on $\epsilon_i$ and $\epsilon_j$. Thus,
\begin{align*}
&\E[(H_i \epsilon_{(i)})(H_j\epsilon_{(j)})\epsilon_i \epsilon_j|E] \\
&= \E[(Q_{ij} + H_{ij} \epsilon_j)(Q_{ji} + H_{ji}\epsilon_{i})\epsilon_i \epsilon_j|E] \\
&= \E(Q_{ij}Q_{ji}|E)\E(\epsilon_i\epsilon_j) + \E(Q_{ij}|E) H_{ji}\E(\epsilon_i^2\epsilon_j) + H_{ij}\E(Q_{ji}|E)\E(\epsilon_i \epsilon_j^2) + H_{ij}H_{ji}\E(\epsilon_i^2\epsilon_j^2)\\
&= \sigma^4 H_{ij}H_{ji}\le \frac{\sigma^4}{2}(H_{ij}^2+H_{ji}^2).
\end{align*}
Similarly, 
\begin{align*}
\E[(H_i \epsilon_{(i)})^2 \epsilon_i^2|E ] &= \E[(H_i \epsilon_{(i)})^2|E]\E[ \epsilon_i^2|E ]=  \sigma^4\|H_i\|^2.
\end{align*}
Using these estimates, we get
\begin{align}
\E[(u-w)^2|E] &\le \frac{\sigma^4}{n}\sum_{i=1}^n \|H_i\|^2 + \frac{\sigma^4}{n}\sum_{i=1}^n \sum_{j\ne i} H_{ij}^2\notag \\
&= \frac{2\sigma^4}{n}\sum_{i=1}^n\|H_i\|^2.\label{uwe}
\end{align}
where $H$ is the matrix whose $(i,j)^{\mathrm{th}}$ element is $H_{ij}$. Now, note that
\begin{align*}
\|H_i\|^2 &= H_i H_i^T \\
&= \tZ_i (\tZ_{(i)}^T \tZ_{(i)})^{-1} \tZ_{(i)}^T \tZ_{(i)} (\tZ_{(i)}^T \tZ_{(i)})^{-1} \tZ_i^T\\
&= \tZ_i (\tZ_{(i)}^T \tZ_{(i)})^{-1} \tZ_i^T.
\end{align*}
Now note that for each $i$, 
\begin{align*}
\|H_i\|^2 &\le \|\tZ_i\|^2 \|(\tZ_{(i)}^T \tZ_{(i)})^{-1}\| \\
&\le \|\tZ_i\|^2 (\|(\tZ_{(i)}^T \tZ_{(i)})^{-1}-(\tZ^T\tZ)^{-1}\| + \|(\tZ^T\tZ)^{-1}\|)\\
&\le \|\tZ_i\|^2(\max_{1\le j\le n}\|(\tZ_{(j)}^T \tZ_{(j)})^{-1}-(\tZ^T\tZ)^{-1}\| + \|(\tZ^T\tZ)^{-1}\|).
\end{align*}
Therefore by Lemma \ref{zlmm1} and Lemma \ref{zlmm},
\begin{align*}
\sum_{i=1}^n \|H_i\|^2 &\le O_P(n^{-1}) \sum_{i=1}^n \|\tZ_i\|^2\\
&= O_P(n^{-1}) \tr(\tZ\tZ^T)\\
&= O_P(n^{-1}) \tr(\tZ^T\tZ) = O_P(n^{-1}) O_P(n) = O_P(1).
\end{align*}
Plugging this into \eqref{uwe}, we get
\begin{align*}
\E[(u-w)^2|E] &\le \frac{2\sigma^4}{n} O_p(1) = O_P(n^{-1}).
\end{align*}
Using this, and the same argument as in the proof of \eqref{ytilde1}, we conclude that $u-w$ tends to zero in probability.
\end{proof}

The following lemma gives the joint distribution of $(u,W)$.
\begin{lmm}\label{uwlmm}
Let $B$ and $D$ be as in Lemma \ref{zlmm}. As $n\to \infty$, $(u,W)$ converges in law to the bivariate normal distribution with mean zero and covariance matrix
\[
\sigma^2 \begin{bmatrix}
\beta_0^T D^TB^{-1}D\beta_0 & D^TB^{-1}D\beta_0\\
\beta_0^T D^TB^{-1}D & \Sigma
\end{bmatrix}.
\]
\end{lmm}
\begin{proof}
Conditional on $E$, the pair $(v, W)$ is jointly normal with mean $0$, and covariance matrix
\[
\frac{\sigma^2}{n}
\begin{bmatrix}
\mu^T\mu & X^T\mu\\
\mu^TX & X^TX
\end{bmatrix}.
\]
The equations \eqref{mulim} and \eqref{xmulim} give the limits of $\frac{1}{n}\mu^T\mu$ and $\frac{1}{n}X^T\mu$, and we know that $\frac{1}{n}X^TX \to \Sigma$. Thus, the above matrix converges in probability to the matrix displayed in the statement of the lemma. By Lemma \ref{uvlmm}, $(u,W)$ must have the same limiting distribution as $(v,W)$. A simple argument via characteristic functions now completes the proof.
\end{proof}
We are now ready to prove Theorem \ref{mainthm}.
\begin{proof}[Proof of Theorem \ref{mainthm}]
Let $B$ and $D$ be as in Lemma \ref{vrlmm}. By Lemma \ref{uwlmm} and Lemma \ref{vrlmm}, 
\[
u-V^T({\textstyle\frac{1}{n}}X^TX)^{-1}W = 
\begin{bmatrix}
1 & - V^T(\frac{1}{n}X^TX)^{-1}
\end{bmatrix}
\begin{bmatrix}
u \\
W
\end{bmatrix}
\]
converges in distribution to a normal random variable with mean zero and variance
\begin{align*}
&
\sigma^2\begin{bmatrix}
1 & - \beta_0^TD^TB^{-1} D\Sigma^{-1}
\end{bmatrix}
\begin{bmatrix}
\beta_0^T D^TB^{-1}D\beta_0 & D^TB^{-1}D\beta_0\\
\beta_0^T D^TB^{-1}D & \Sigma
\end{bmatrix}
\begin{bmatrix}
1\\
- \Sigma^{-1} D^TB^{-1} D\beta_0
\end{bmatrix}\\
&=\sigma^2(\beta_0^T D^TB^{-1}D\beta_0 - 2\beta_0^TD^TB^{-1} D\Sigma^{-1} D^TB^{-1}D \beta_0 \\
&\qquad \qquad + \beta_0^TD^TB^{-1} D\Sigma^{-1} \Sigma \Sigma^{-1} D^TB^{-1} D\beta_0)\\
&= \sigma^2(\beta_0^T D^T B^{-1}D\beta_0 - \beta_0^TD^TB^{-1} D\Sigma^{-1} D^TB^{-1} D\beta_0).
\end{align*}
But note that by Lemma \ref{vrlmm}, the quantity within the parentheses is equal to the reciprocal of the limit of $r$. By Lemma~\ref{mainlmm}, this completes the proof of the theorem.
\end{proof}

\section{Derivation of ALO for the Logistic Lasso Problem}\label{appendix:ALOlogistic}
First, the objective function is as follows: \begin{equation*}
        L(\beta) = \sum_{i=1}^n y_i \log p_i + (1-y_i)\log (1-p_i) +\lambda ||\beta||_1
\end{equation*}
where $p_i = \frac{1}{1+e^{-z_i^\top \beta}}$.
    
Then based on the formula given by \citeasnoun**{auddy2024approximate}, we need to compute:
\begin{equation*}
    \begin{aligned}
    &l_i(\hat{\beta})=l(y_i|z_i^\top \hat{\beta})=y_i \log p_i + (1-y_i)\log (1-p_i)\\
    &\dot{l}_i(\hat{\beta})=(\frac{y_i}{p_i}- \frac{1-y_i}{1-p_i})\frac{\partial p_i}{\partial z_i^\top \hat{\beta}} = \frac{y_i-p_i}{p_i(1-p_i)}\frac{e^{-z_i^\top \hat{\beta}}}{\left(1+e^{-z_i^\top \hat{\beta}}\right)^2} = y_i-p_i\\
    &\ddot{l}_i(\hat{\beta})=-\frac{e^{-z_i^\top \hat{\beta}}}{\left(1+e^{-z_i^\top \hat{\beta}}\right)^2}=-p_i(1-p_i)
    \end{aligned}
\end{equation*}

Therefore,
\begin{equation*}
    \log \frac{p_{-i}}{1-p_{-i}}= z_i^\top \hat{\beta} - \frac{y_i-p_i}{p_i(1-p_i)}\frac{H_{ii}}{1-H_{ii}}
\end{equation*}
where $H = Z_\mathcal{S} (Z_\mathcal{S}^\top W Z_\mathcal{S})^{-1}Z_\mathcal{S}^\top W$, $\mathcal{S} = \{i:\hat{\beta}_i\neq 0\}$ is the active set of regular Lasso estimates, without leave-one-out cross fitting.

\end{document}